%% file: main.tex
\documentclass[11pt]{article}

\usepackage{titling}
\usepackage{blindtext}

\usepackage[affil-it]{authblk}

\usepackage{ifthen}
\usepackage{fullpage}
\usepackage{natbib}
\usepackage{amsmath}
\usepackage{amssymb}
\usepackage{hyperref}
\usepackage{amsthm}
\usepackage{verbatim}
\usepackage{enumerate}
\usepackage{cleveref}

\usepackage{tikz}
\usepackage{pgfplots}
\usetikzlibrary{patterns}
\usepgfplotslibrary{fillbetween}
\usetikzlibrary{intersections}

\include{notation}

\include{math}

\theoremstyle{plain}
\newtheorem{theorem}{Theorem}[section]
\newtheorem{lemma}[theorem]{Lemma}

\newtheorem{proposition}[theorem]{Proposition}
\newtheorem{corollary}[theorem]{Corollary}

\newenvironment{numberedtheorem}[1]{%
\renewcommand{\thetheorem}{#1}%
\begin{theorem}}{\end{theorem}\addtocounter{theorem}{-1}}

\theoremstyle{plain}
\newtheorem{definition}{Definition}[section] 

\title{An End-to-end Argument in Mechanism Design\\
(Prior-independent Auctions for Budgeted Agents)
\thanks{Appearing at the 2018 IEEE Symposium on Foundations of Computer Science.}}

\author{Yiding Feng
  \thanks{Department of Electrical Engineering and Computer Science, Northwestern University. Email: \texttt{yidingfeng2021@u.northwestern.edu}.}
  \and
Jason D. Hartline
  \thanks{Department of Electrical Engineering and Computer Science, Northwestern University.
  Email: \texttt{hartline@northwestern.edu}.}}

\date{}
\begin{document}

\maketitle

\input{abstract}

\input{intro}
\input{prelim}
\input{clinching}

\input{optimal_DSIC}

\input{revelation_gap}
\input{irregular}
\input{winner}

\input{revenue}

\bibliographystyle{apalike}
\bibliography{auctions}

\appendix

\input{appendix-price-jump}

\input{appendix-winner}

\end{document}

%% file: notation.tex
\usepackage{amssymb}
\usepackage{amsmath}
\usepackage{ifthen}
\usepackage{fixmath}
\usepackage{xfrac}
\usepackage{sidecap}
\usepackage{subfig}
\usepackage{caption}

\newcommand{\rand}{\pi}

\newcommand{\revelationgap}{\gaplowerbound}

\DeclareMathAlphabet{\mathpzc}{OT1}{pzc}{m}{it}

\newcommand{\agind}[1][i]{_{#1}}

\newcommand{\ironed}{\bar}
\newcommand{\constrained}{\hat}
\newcommand{\optconstrained}{\composed{\optimized}{\constrained}}
\newcommand{\optimized}[1]{#1\opt}
\newcommand{\differentiated}[1]{#1'}

\newcommand{\tagged}[2]{{#2}^{#1}}

\newcommand{\starred}[1]{#1^\star}
\newcommand{\primedarg}[1]{#1\primed}
\newcommand{\noaccents}[1]{#1}
\newcommand{\composed}[3]{#1{#2{#3}}}

\newcommand{\newagentvar}[3][\noaccents]{%
\expandafter\newcommand\expandafter{\csname #2\endcsname}{#1{#3}}%
\expandafter\newcommand\expandafter{\csname #2s\endcsname}{#1{\boldsymbol{#3}}}%
\expandafter\newcommand\expandafter{\csname #2smi\endcsname}[1][i]{#1{\boldsymbol{#3}}_{-##1}}%
\expandafter\newcommand\expandafter{\csname #2i\endcsname}[1][i]{#1{#3}\agind[##1]}%
\expandafter\newcommand\expandafter{\csname #2ith\endcsname}[1][i]{#1{#3}_{(##1)}}%
}

\newcommand{\newitemvar}[3][\noaccents]{%
\expandafter\newcommand\expandafter{\csname #2\endcsname}{#1{#3}}%
\expandafter\newcommand\expandafter{\csname #2s\endcsname}{#1{\boldsymbol{#3}}}%
\expandafter\newcommand\expandafter{\csname #2smj\endcsname}[1][j]{#1{\boldsymbol{#3}}_{-##1}}%
\expandafter\newcommand\expandafter{\csname #2j\endcsname}[1][j]{#1{#3}_{##1}}%
\expandafter\newcommand\expandafter{\csname #2jth\endcsname}[1][j]{#1{#3}_{(##1)}}%
}

\newagentvar{alloc}{x}

\newagentvar{falloc}{z}

\newagentvar{quant}{q}
\newagentvar[\constrained]{exquant}{\quant}
\newagentvar[\constrained]{critquant}{\quant}  
\newagentvar[\optconstrained]{monoq}{\quant}
\newagentvar{Val}{\nu}
\newagentvar{toquant}{Q}

\newagentvar{qprice}{\price}
\newagentvar{qrev}{R}
\newagentvar[\ironed]{iqrev}{\qrev}

\newagentvar{qalloc}{y}
\newagentvar{cumalloc}{Y}
\newagentvar[\constrained]{cumcalloc}{\cumalloc}
\newagentvar[\ironed]{ialloc}{\qalloc}
\newagentvar[\ironed]{icumalloc}{\cumalloc}

\newcommand{\exposted}[1]{#1^{\text{\it EP}}}
\newagentvar[\composed{\exposted}{\constrained}]{excalloc}{\qalloc}
\newagentvar[\exposted]{exalloc}{\qalloc}
\newagentvar[\exposted]{extalloc}{\talloc}
\newagentvar[\exposted]{exfalloc}{\falloc}

\newagentvar{rev}{R}
\newagentvar[\differentiated]{marg}{\rev}
\newagentvar{rawrev}{P}
\newagentvar[\differentiated]{rawmarg}{\rawrev}

\newagentvar[\primedarg]{inducedrev}{\rev}

\newagentvar[\tilde]{pseudorev}{\rev}
\newagentvar[\tilde]{pseudorawrev}{\rawrev}

\newagentvar{typespace}{{\cal T}}
\newagentvar{typesubspace}{S}

\newagentvar{type}{t}
\newagentvar{othertype}{s}
\newagentvar{val}{v}
\newagentvar{highestval}{h}
\newagentvar{lowestval}{0}
\newagentvar{cumval}{V}
\newagentvar{cumprice}{P}
\newagentvar{welcurve}{W}
\newagentvar{revcurve}{R}
\newagentvar{outcome}{w}
\newagentvar{outcomespace}{{\cal W}}

\newcommand{\served}[1]{#1^1}
\newcommand{\nonserved}[1]{#1^0}
\newcommand{\alloced}[1]{#1^{\alloc}}
\newcommand{\allocedi}[1]{#1^{\alloci}}
\newagentvar[\alloced]{xoutcome}{\outcome}
\newagentvar[\allocedi]{xioutcome}{\outcome}
\newagentvar[\served]{soutcome}{\outcome}

\newagentvar[\nonserved]{nsoutcome}{\outcome}

\newagentvar{price}{p}
\newagentvar{talloc}{\alloc}
\newagentvar[\tilde]{balloc}{\alloc}
\newagentvar[\tilde]{bprice}{\price}

\newagentvar{act}{a}
\newagentvar{bidspace}{A}
\newagentvar{actspace}{A}

\newagentvar[\constrained]{critval}{\val}
\newagentvar[\constrained]{crittype}{\type}
\newagentvar[\constrained]{critvirt}{\virt}
\newagentvar[\constrained]{reserve}{\val} 
\newagentvar[\constrained]{bidreserve}{\bid} 
\newagentvar[\optconstrained]{monop}{\val}
\newagentvar[\constrained]{monot}{\type}
\newagentvar[\optimized]{monorev}{\rev}

\newagentvar{rcalloc}{y}
\newagentvar[\optimized]{optrcalloc}{\rcalloc}
\newagentvar{biddist}{G}

\newagentvar[\constrained]{critbid}{\bid}
\newagentvar[\constrained]{cbid}{B}
\newagentvar[\primedarg]{wbid}{\bid}
\newagentvar{gfunc}{\vartheta}
\newagentvar{block}{C}

\newagentvar{util}{u}

\newagentvar{strat}{s}
\newagentvar{bid}{b}

\newagentvar{virt}{\phi}
\newagentvar{cumvirt}{\Phi}
\newagentvar{qvirt}{\phi}
\newagentvar[\ironed]{ivirt}{\virt}
\newagentvar[\ironed]{icumvirt}{\cumvirt}

\newagentvar[\tagged{\text{SD}}]{sdvirt}{\virt}
\newagentvar[\composed{\tagged{\text{SD}}}{\ironed}]{sdivirt}{\virt}
\newagentvar[\tagged{\text{MD}}]{mdvirt}{\virt}
\newagentvar[\composed{\tagged{\text{MD}}}{\ironed}]{mdivirt}{\virt}

\newagentvar{dist}{F}
\newagentvar{dens}{f}
\newagentvar{hazard}{h}
\newagentvar{cumhazard}{H}

\newagentvar[\ironed]{iprice}{\price}  
\newagentvar[\ironed]{ival}{\val}  

\newagentvar[\ironed]{icumval}{\cumval}
\newagentvar{ints}{{\cal I}}

\newagentvar{wal}{w}

\newitemvar{pos}{j}
\newitemvar{weight}{w}
\newitemvar[\differentiated]{mweight}{\weight}
\newitemvar{udtype}{\type}
\newitemvar[\constrained]{udcrittype}{\type}
\newitemvar{udalloc}{\alloc}
\newitemvar{udprice}{\price}
\newitemvar[\constrained]{udcalloc}{\qalloc}
\newitemvar[\constrained]{udcumcalloc}{\cumalloc}

\newagentvar{mech}{{\cal M}}
\newagentvar[\skew{5}{\hat}]{cmech}{{\cal M}}
\newagentvar{alg}{{\cal A}}

\newcommand{\Obj}[1]{\mathbf{Payoff}[#1]}

\newcommand{\lagrange}{\lambda}

\newcommand{\budget}{B}

\DeclareMathOperator{\OPT}{OPT}

\newcommand{\reals}{{\mathbb R}}

\newagentvar{trans}{\sigma}

\newagentvar{demandset}{S}

\newcommand{\opt}{^{\star}}
\newcommand{\primed}{^\dagger}
\newcommand{\doubleprimed}{^{\ddagger}}

\newagentvar{distout}{w}
\newagentvar{idistout}{\bar{w}}

\newcommand{\highval}{h}

\newagentvar{gap}{\delta}

\newagentvar[\starred]{calloc}{\alloc}
\newcommand{\welfare}{\text{\textbf{welfare}}}

\newcommand{\PO}{\text{PP}}
\newcommand{\allocPO}{\alloc_\PO}

\newcommand{\CL}{\text{Clinch}_k}
\newcommand{\allocCL}{\alloc_{\CL}}

\newcommand{\aalloc}{z}
\newcommand{\event}{\mathcal{E}}

\newcommand{\valp}{\val\primed}
\newcommand{\valdp}{\val\doubleprimed}
\newcommand{\dvp}{\lambda(\valdp - 1)}

\newcommand{\mt}{Z}
\newcommand{\allocspace}{\mathcal{X}}
\newcommand{\MECHS}{\text{MECHS}}
\newcommand{\DISTS}{\text{DISTS}}
\newcommand{\gaplowerbound}{1.013}
\newcommand{\clinchinglowerbound}{1.03}

\newcommand{\lvalue}{\mathbf L}
\newcommand{\mvalue}{\mathbf M}
\newcommand{\hvalue}{\mathbf H}
\newcommand{\LS}{\lvalue*}
\newcommand{\MM}{\mvalue\mvalue}
\newcommand{\HM}{\hvalue\mvalue}
\newcommand{\MH}{\mvalue\hvalue}
\newcommand{\HH}{\hvalue\hvalue}

\newcommand{\lambdas}{\mathbold \lambda}
\newcommand{\Lambdas}{\mathbf \Lambda}
\newcommand{\betas}{\mathbold \beta}
\newcommand{\mus}{\mathbold \mu}

\newcommand{\Sp}{S\primed}
\newcommand{\Sdp}{S\doubleprimed}

\newcommand{\kp}{k\primed}
\newcommand{\kdp}{k\doubleprimed}

%% file: math.tex
\DeclareMathOperator{\argmax}{argmax}

\newcommand{\given}{\,\mid\,}

\newcommand{\prob}[2][]{\text{\bf Pr}\ifthenelse{\not\equal{}{#1}}{_{#1}}{}\!\left[{\def\givenn{\middle|}#2}\right]}
\newcommand{\expect}[2][]{\text{\bf E}\ifthenelse{\not\equal{}{#1}}{_{#1}}{}\!\left[{\def\givenn{\middle|}#2}\right]}

\newcommand{\tparen}{\big}
\newcommand{\tprob}[2][]{\text{\bf Pr}\ifthenelse{\not\equal{}{#1}}{_{#1}}{}\tparen[{\def\given{\tparen|}#2}\tparen]}
\newcommand{\texpect}[2][]{\text{\bf E}\ifthenelse{\not\equal{}{#1}}{_{#1}}{}\tparen[{\def\given{\tparen|}#2}\tparen]}

\newcommand{\sprob}[2][]{\text{\bf Pr}\ifthenelse{\not\equal{}{#1}}{_{#1}}{}[#2]}
\newcommand{\sexpect}[2][]{\text{\bf E}\ifthenelse{\not\equal{}{#1}}{_{#1}}{}[#2]}

%% file: abstract.tex
\begin{abstract}
This paper considers prior-independent mechanism design, namely
identifying a single mechanism that has near optimal performance on
every prior distribution.  We show that mechanisms with truthtelling
equilibria, a.k.a., revelation mechanisms, do not always give optimal
prior-independent mechanisms and we define the revelation gap to
quantify the non-optimality of revelation mechanisms.  This study
suggests that it is important to develop a theory for the design of
non-revelation mechanisms.

Our analysis focuses on welfare maximization in single-item auctions
for agents with budgets and a natural regularity assumption on their
distribution of values.  The all-pay auction (a non-revelation
mechanism) is the Bayesian optimal mechanism; as it is
prior-independent it is also the prior-independent optimal mechanism
(a 1-approximation).  We prove a lower bound on the prior-independent
approximation of revelation mechanisms of $\gaplowerbound$ and that
the clinching auction (a revelation mechanism) is a prior-independent
$e\approx 2.714$ approximation.  Thus the revelation gap for
single-item welfare maximization with public budget agents is in
$[\gaplowerbound,e]$.

Some of our analyses extend to the revenue objective, position
environments, and irregular distributions.
\end{abstract}

%% file: intro.tex
\section{Introduction}

%
%
The {\em end-to-end principle} in distributed systems advocates
environment-independent protocols (for the center) that push
environment-specific complexity to the applications (the end points)
that use the protocol \citep{SRC-84}.  This principle enabled the
Internet protocols designed for the workloads of the 1980s to continue
to succeed with workloads of the 2010s.  On the other hand, research
in mechanism design (which governs the design of protocols for
strategic agents and has application both in computer science and
economics) almost exclusively adheres to the {\em revelation
  principle} which suggests the design of mechanisms where each
agent's best strategy is to truthfully report her preferences.  In
revelation mechanisms the agents (the end points) have very a simple
``report your true preference'' strategies and the mechanism (the
center) has the complex task of finding an outcome that both enforces
this truthfulness property and also obtains a desirable outcome.
Unsurprisingly, optimal revelation mechanisms tend to be complex and
dependent on the environment.  This paper demonstrates that the
end-to-end argument has bite in mechanism design by showing that
non-revelation mechanisms are strictly better than revelation
mechanisms for a canonical mechanism design problem.

%
%
In prior-independent mechanism design, it is assumed that the agents'
preferences are drawn from a distribution that is not known to the
designer.  The goal of prior-independent mechanism design is to
identify mechanisms that are good approximations to the optimal
mechanism that is tailored to the distribution of preferences.
Specifically, a mechanism is sought to minimize the ratio of its
expected performance to the expected performance of the optimal
mechanism in worst case over distributions from which the preferences
of the agents are drawn.  This notion is a standard one that has been
applied to revenue maximization \citep{DRY-10,RTY-12,FILS-15,AB-18},
multi-dimensional mechanism design \citep{DHKT-11,RTY-15}, makespan
minimization \citep{CHMS-13}, mechanism design for risk-averse agents
\citep{FHH-13}, and mechanism design for agents with interdependent
values \citep{CFK-14}.  In none of these scenarios is the optimal
prior-independent mechanism known; cf.\ \citet{FILS-15} and
\citet{AB-18}.

%
%
The revelation principle suggests that if there is a mechanism with a
good equilibrium outcome, there is a mechanism where truthtelling
achieves the same outcome in a truthtelling equilibrium.  Due to the
revelation principle, much of the theory of mechanism design is
developed for revelation mechanisms, i.e., ones where truthtelling is
an equilibrium.  The proof of the revelation principle is simple: A
revelation mechanism can simulate the equilibrium strategies in the
non-revelation mechanism to obtain the same outcome as a truthtelling
equilibrium, i.e., agents report true preferences to the revelation
mechanism, it simulates the agent strategies in the non-revelation
mechanism, and it outputs the outcome of the simulation.  For Bayesian
non-revelation mechanisms (where the agents' preferences are drawn
from a prior distribution), the agents' equilibrium strategies are a
function of the prior and thus the corresponding revelation mechanism
constructed via the revelation principle is not prior-independent.
Thus, the restriction to revelation mechanisms is not generally
without loss for prior-independent mechanism design.  Non-revelation
mechanisms, on the other hand, are widely used in practice and
frequently have easily to identify and natural equilibria (e.g., in
rank-based auctions, see \citealp{CH-13}).  Our proof of a non-trivial
revelation gap -- that the prior-independent approximation factor of
the best non-revelation mechanism is better than that of the best
revelation mechanism -- gives concrete motivation for a theory of
mechanism design without the revelation principle.

It is not hard to invent pathological scenarios where there is a
non-trivial revelation gap.  Instead, this paper considers the
canonical environment of welfare maximization for agents with budgets
and shows such a gap even for distributions on preferences that
satisfy a standard regularity property.  Moreover, the environment in
which we exhibit the revelation gap suggests the end-to-end principle:
the agents can easily implement the optimal outcome in the equilibrium
of a simple mechanism, while revelation mechanisms that satisfy the
constraints must be complex and either prior-dependent or non-optimal.

\paragraph{Main Results.}
%
%
Our analysis focuses on welfare maximization in a canonical
single-item environment with ex ante symmetric budget constrained
agents, i.e., each agent's value is drawn independently and
identically from an unknown distribution and the agent cannot make
payments that exceed a known and identical budget \citep[cf.][]{mas-00}.  Our
main treatment of this problem will make a simplifying assumption that
the distribution follows a regularity property that implies that the
Bayesian optimal mechanism has a nice form \citep{PV-14}.  Our results
require a symmetric environment, i.e., an i.i.d.\ distribution and
identical budget.  A number of our results extend to the objective of
revenue \citep{LR-96}, to position environments as popularized as a
model for ad auctions \citep{DHH-13}, and beyond regular distributions
\citep{DHH-13}.  For clarity the main results are described first for
welfare maximization, single-item environments, and regularly
distributed agent values.

The main challenge in demonstrating a revelation gap is that it is
difficult to identify prior-independent optimal mechanisms,
cf.\ \citet{FILS-15}.  Though the question has been considered, the
prior literature has no examples of optimal prior-independent
mechanisms for non-trivial environments. \footnote{Contemporaneously
  with our results, \citet{AB-18} showed that the second-price auction
  is the prior-independent optimal revelation mechanism for revenue
  when the agents' values are distributed according to a monotone
  hazard rate distribution.} Our non-trivial revelation-gap theorem
follows from three results.  First, the all-pay auction (from the
literature, defined below) has a unique equilibrium that is Bayesian
optimal and it is prior-independent.  Second, we obtain a lower bound
on the ability of a prior-independent revelation mechanism to
approximate the Bayesian optimal mechanism by identifying the dominant
strategy incentive compatible mechanism that is Bayesian optimal for
the uniform distribution.  The performance of this mechanism is
strictly worse than that of the Bayesian optimal mechanism (which is
Bayesian incentive compatible); specifically the gap is
$\revelationgap$. Third, we show that the dominant strategy incentive
compatible clinching auction (from the literature, defined below) is
an $e\approx 2.72$ approximation to the Bayesian optimal mechanism.
Combining the upper and lower bounds we see a revelation gap between
$\revelationgap$ and $e$. \footnote{To better appreciate the magnitude
  of this lower bound, notice that it is demonstrated for two agents
  with uniformly distributed values where the optimal expected welfare (even
  without budgets) is is $2/3$ and the lottery mechanism (which gives
  the item to a random agent) has expected welfare $1/2$ and is a $4/3
  \approx 1.33$ approximation.}  The first result follows naturally
  from the literature; the second and third results are the main
  technical contributions of the paper.

Three auctions are at the forefront of our study.  The {\em all-pay
  auction} solicits bids, assigns the item to the highest agent, and
charges all agents their bids.  The {\em clinching auction}
\citep{aus-04,DLN-08,GML-15} is an ascending price auction that can be
thought of as allocating a unit measure of lottery tickets: a price is
offered in each stage, each agent specifies the measure of tickets
desired at the given price, each agent is allocated a number of
tickets that is equal to the minimum of her demand and the measure of
remaining tickets if this agent is only allowed to buy tickets after
all other agents have bought as much as they desire
first. \footnote{For example, at a price of 0 all agents would want to
  buy all the tickets, but the agent that arrives last gets no
  tickets, thus no agents get any tickets at this price; the price
  increases.}  The {\em middle-ironed clinching auction} -- which we
identify as the optimal dominant strategy incentive compatible
mechanism -- behaves like the clinching auction except that values
that fall within a middle range are ironed. The allocation that an
agent in this middle range receives is the average over he original
allocation of for middle range values in the clinching auction.  This
averaging results in the the budget binding later and more efficient
outcomes than in the original clinching auction.

The second step, mentioned above, is to obtain a lower bound on the
prior-independent approximation of a revelation mechanism.  Our
analysis begins with the observation that a prior-independent
revelation mechanism must be Bayesian incentive compatible for every
distribution.  For two agents, this condition is equivalent to being
dominant strategy incentive compatible.  We ask whether there a gap
between the Bayesian optimal dominant strategy and Bayesian incentive
compatible mechanism.  The comparison between optimal dominant
strategy and Bayesian incentive compatible mechanism is standard for
multi-dimensional mechanism design problems, e.g., see
\citet{GGKMS-13} and \citet{yao-17}; we are unaware of previous
studies of this phenomenon for single-dimensional agents with
non-linear preferences.  We answer this question positively by writing
the dominant strategy mechanism design problem as a linear program and
solving it by identifying a dual solution that proves the optimality
of the middle-ironed clinching auction, cf.\ \citet{PV-14} and
\citet{DW-17}.  The identified gap gives a lower bound on the
approximation factor of the optimal prior-independent mechanism.

The third step, mentioned above, proves that the prior-independent
approximation factor of the clinching auction auction is at most $e$
and resolves in the affirmative an open question from \citet{DHH-13}.
Our proof follows from a novel adaptation of a standard method for
approximation results in mechanism design where an auction's
performance is compared to the upper bound given by the ex ante
relaxation, in this case, the welfare of the optimal mechanism that
sells one item in expectation over the random draws of the agents'
values (i.e., ex ante) rather than for all draws of the agents' values
(i.e., ex post).  This method was introduced by \citet{CHK-07},
formalized by \citet{ala-11,ala-14}, generalized by \citet{AFHH-13}, and
employed in many subsequent analyses.

\paragraph{Extensions.}
A number of our results extend beyond regular distributions,
single-item environments, and the welfare objective as described
above.  These extensions all require that the environment be
symmetric, specifically, that the agents' values are independent and
identically distributed and their budgets are identical.

For irregular distributions the welfare-optimal auction is not
generally the all-pay auction; moreover, it does not generally have a
prior-independent implementation.  We prove that the all-pay auction
is a prior-independent two approximation.  Both the regular and
irregular prior-independent optimality and approximation results for
the one-item all-pay auction extend to the all-pay position auction.

The degradation of the approximation factor by a factor of two for
irregular distributions extends to the single-item clinching auction
which is an $e$ approximation for regular distributions (as described
above) and a $2e$ approximation for irregular distributions.

For the revenue objective \citep{LR-96}, with appropriate definition
of regularity, the $n$-agent single-item all-pay auction is a
prior-independent $n/(n-1)$ approximation to the revenue optimal
auction \citep[cf.][]{BK-96}.


\paragraph{Important Directions.}

The most general direction suggested by this work is for a systematic
development of non-revelation mechanism design.  Unfortunately, it is
not generally helpful to do revelation mechanism design and then try
to go from the suggested revelation mechanism to a practical and
simple non-revelation mechanism.  There is a nascent literature on this
topic.  Papers working to develop a theory of non-revelation mechanism
design include \cite{CH-13}, which proves the uniqueness and
optimality of equilibria in symmetric rank-based auctions;
\citet{CHN-14,CHN-16}, which gives data driven methods for optimizing
non-revelation mechanisms in symmetric environments; and
\citet{HT-16}, which gives a theory for non-revelation sample
complexity and the design of approximately optimal non-revelation
mechanisms in asymmetric environments.

While the literature has many interesting approximation bounds for
prior-independent mechanism design.  Rarely have the prior-independent
optimal mechanism been identified.  Moreover, the prior-independent
approximation factors achievable tend to be surprising; for example,
\citet{FILS-15} show that the second-price action is not the optimal
prior-independent mechanisms for two-agent revenue maximization with
agents with regularly distributed values. \footnote{\citet{AB-18} show
  that with more restrictive monotone hazard rate distributions, the
  second-price auction is an optimal prior-independent revelation
  mechanism.}  The literature lacks general techniques for answering
this question.

We have observed that there is a very simple prior-independent optimal
mechanism for welfare maximization in symmetric environments for
agents with identical budgets.  This mechanism, namely the all-pay
auction, achieves its optimal outcome in Bayes-Nash equilibrium.  The
general question of identifying prior-independent non-revelation
mechanisms that optimize a desired objective, like welfare or revenue,
needs to be asked with care.  Without restrictions to this question,
it is asked and answered in the literature on non-parametric
implementation theory, see the survey of \citet{J-01}.  This
literature shows that arbitrarily close approximations, called
``virtual implementations'', to the Bayesian optimal mechanism can be
implemented by an uninformed principal.  The mechanisms in this
literature tend to be sequential -- where agents interact in multiple
rounds -- and require agents to make reports about their own
preferences and crossreports about their beliefs on other agents'
preferences.  Our perspective on these results is that they take the
model of Bayes-Nash equilibrium too literally and the resulting
cross-reporting mechanisms are both fragile and impractical.  One
approach for ruling out these mechanisms is to restrict attention to
mechanism formats that are commonly occurring in practice.
Specifically, in the general {\em winner-pays-bid} format: agents bid,
an allocation rule maps bids to a set of winners, and all winners pay
their bids; in the general {\em all-pay} format: agents bid, an
allocation rule maps bids to a set of winners, and all agents pay
their bids.  There may be other restricted formats that are also
interesting for specific scenarios, e.g., the seller-offer mechanisms
that are prevalent as real estate exchange mechanisms \citep{NYK-14}.

Finally, there are still many gaps in our understanding of auctions
for identically distributed agents with common budgets.  For welfare,
the bounds in this paper show that the clinching auction's
approximation factor for the welfare objective is in
$[\clinchinglowerbound,e]$ for regular distributions and $[2,2e]$ for
irregular distribution.  Sharpening these bounds is an open question.
Moreover, we conjecture that the clinching auction is also a
prior-independent constant approximation for the revenue objective
(with i.i.d.\ public-budget regular distributions).  We also leave
open a number of questions with regard to the Bayesian optimal
dominant strategy incentive compatible mechanism for agents with
budgets.  We conjecture that the welfare optimality of the
middle-ironed clinching auction extends from uniform distributions to
regular distributions.  We leave open the question of a similar result
for the revenue objective, even for the special case of uniform
distributions.  There are specific issues that prevent straightforward
generalization of our approach of using the dual to certify the
optimality of the middle-ironed clinching auction for these questions.

\paragraph{Other Related Work.}
There is a significant area of research analyzing the performance of
simple non-revelation mechanisms in equilibrium, a.k.a., the {\em
  price of anarchy}.  Generally these
mechanisms are prior-independent and the aim of the literature, e.g.\ \citet{ST-13}, is to demonstrate that  they are approximately
efficient.  On the other hand, for welfare maximization in many of the studied environments, there is a DSIC revelation mechanism that is (exactly) efficient and, thus, there is
no revelation gap.  Though this literature focuses on the analysis rather than the design of
mechanisms, two conclusions for mechanism design
are: (a) that a simple {\em revenue covering property} is sufficient
\citep{HHT-14}, necessary \citep{DK-15}, and potentially optimizable; and (b) that this property
(and also a more general {\em smoothness property}) is closed under composition, i.e., when multiple
independent mechanisms are run simultaneously \citep{ST-13}.  For a surveys of these
and other results see \citet{RST-17} and Chapter~6 of \citet{har-16}.

For agents with budgets, approximation mechanisms have been studied
from both a design and analysis perspective for the {\em liquid
  welfare} benchmark proposed by \citet{CMM-11}, \citet{ST-13}, and
\citet{DP-14}.  The liquid welfare benchmark is the optimal surplus of
a feasible allocation when each agent's contribution to the surplus is
the maximum of her value for her allocation and her budget.  These and
subsequent papers show that simple mechanisms have welfare that
approximate the liquid welfare benchmark.  Unfortunately, when
evaluated under the formal study of benchmarks for mechanism design
developed by \citet{HR-08} and summarized in Chapter~7 of
\citet{har-16}, the liquid welfare does not satisfy a key property.
Specifically, there are i.i.d.\ distributions where the expected
welfare of the Bayesian optimal mechanism is arbitrarily larger than
the expected optimal liquid welfare.  This bound follows because
liquid welfare is at most the sum of the agent budgets which can be
arbitrarily close to zero and, in such cases, is unrelated to the
welfare possible by a mechanism.  \footnote{For example, in a
  single-item environment with budgets identically equal to zero and
  agent values identically equal to one (trivially an
  i.i.d.\ distribution); the lottery, which allocates the item to a
  random agent for free, has welfare one while the liquid welfare is
  zero.}  Thus, testing mechanisms for their worst-case approximation
factor with respect to liquid welfare does not necessarily separate good mechanisms from bad mechanisms.

\paragraph{Organization.}  
In \Cref{s:prelim} we give the preliminaries of our setup. 
In \Cref{s:clinching} we analyze the prior-independent approximation
factor of the clinching auction for public-budget regular agents.  In
\Cref{s:DSIC} we derive the Bayesian optimal DSIC mechanism for two
agents with value uniformly distributed and show it is the
middle-ironed clinching auction.  In \Cref{s:gap}, we prove that the
revelation gap is a constant.  In \Cref{s:irregular}, we analyze the
prior-independent approximation factor of the all-pay auction and the
clinching auction for irregular agents.  In \Cref{s:winner}, we
analyze the approximation ratio of winner-pays-bid mechanisms.  In
\Cref{s:revenue}, we analyze the prior-independent revenue
approximation of the all-pay auction for public-budget regular agents.
In \Cref{app:price-jump}, we give an ascending implementation of the
clinching auction with price jumps (a generalization of the
middle-ironed clinching auction).  In \Cref{app:winner}, we give a
geometric framework for deriving Bayesian optimal mechanisms for
agents with budgets and apply this framework to solve for the Bayesian
optimal winner-pays-bid mechanism.

%% file: prelim.tex
\section{Preliminaries}
\label{s:prelim}

\paragraph{Model for auctions with budgets}
We consider mechanisms for agents with independent and
identically distributed values and identical public budgets.  The budget is
denoted by $\budget$.  For allocation $\alloc \in [0,1]$ and payment
$\price \in \reals$, an agent with value $\val \in \reals$ has utility
$\val \alloc - \price$ if $\price$ is at most the budget $\budget$ and
utility $-\infty$ otherwise.  In other words, the agent cannot under
any circumstances pay more than her budget.
The agents' values are drawn independently and identically from distribution $F$ with support $[0,\highestval]$. 

Denote the strategy function of an agent by
$\strat(\cdot)$ where $\strat(\val)$ is the bid made by the agent
when her value is $\val$.  
A bid profile is $\bids =
(\bidi[1],\ldots,\bidi[n])$.  A mechanism is given by mapping from
bids to allocations and payments which we will denote by
$\ballocs(\bids) = (\balloci[1](\bids),\cdots,\balloci[n](\bids))$ and
$\bprices(\bids) = (\bpricei[1](\bids),\cdots,\bpricei[n](\bids))$.
The outcome of the mechanism $(\ballocs,\bprices)$ and strategy profile
$\strats = (\strati[1],\ldots,\strati[n])$ on a profile of agent
values $\vals$ is denoted by allocation rule $\allocs(\vals) =
\ballocs(\strats(\vals))$ and payment rule $\prices(\vals) =
\bprices(\strats(\vals))$.

The auction designer typically faces a feasibility constraint that
restricts the allocations that can be produced.  For example, a
single-item auction requires that the number of agents allocated is at
most one, i.e., $\sum_i \alloci(\vals) \leq 1$.  A position
environment generalizes a single item auction and is given by a
sequence of decreasing probabilities $\wals =
(\wali[1],\ldots,\wali[n])$ where without loss of generality the
number of positions is equal to the number of agents.  The probability
that an agent is allocated if assigned to position $j$ is $\wali[j]$.
A mechanism then can assign agents to positions (deterministically or
randomly) and this process and the position probabilities induce
allocation probabilities $\alloci(\vals)$.

\paragraph{Basic auction theory}

A Bayes-Nash equilibrium (BNE) in the mechanism $(\ballocs, \bprices)$ 
is a profile of agent strategies $\strats$ where each $\strati[i]$ maps
a value to a bid that is a best response to the other strategies and the 
common knowledge that values are drawn i.i.d.\ from distribution $\dist$. 
The mechanism $(\ballocs,\bprices)$ and strategy profile $\strats$ 
induce for each agent an interim
allocation rule $\alloci(\val) = \expect[\valsmi]{
\alloci(\val,\valsmi)}$.  We will consider only symmetric
distributions and symmetric auctions.  In such auctions, \citet{CH-13}
show that all equilibria are symmetric, thus it is without loss to
drop the subscript and refer to the interim allocation rule and
payment rule as $(\alloc,\price)$.  The \citet{mye-81}
characterization of BNE requires that (a)
the interim allocation $\alloc(\val)$ is monotone non-decreasing and (b) 
the interim payment $\price(\val) =
\val\cdot\alloc(\val) - \int_0^{\val} \alloc(t)\,dt$.  Condition (b) is known as the payment
identity.
A mechanism is Bayesian incentive compatible (BIC) if it induces a 
BNE where each agents strategy is reporting her value truthfully.
A mechanism is interim individual rational (IIR) if the interim allocation is 
non-negative for all value.

A dominant strategy equilibrium (DSE) in the mechanism $(\ballocs, \bprices)$ 
is a profile of agent strategies $\strats$ where each $\strati[i]$ maps
a value to a bid that is a best response regardless of what other agents are doing. 
The characterization of DSE follows from the BNE characterization:
(a) the allocation $\alloc(\vali[i],\valsmi)$ is monotone non-decreasing in $\vali[i]$ and
(b) the payment $\price(\vali[i],\valsmi) = \vali[i]\cdot\alloc(\vali[i],\valsmi) 
- \int_0^{\vali[i]}\alloc(t,\valsmi)dt$. 
A mechanism is dominant strategy incentive compatible (DSIC) if it induces a 
DSE where each agents strategy is reporting her value truthfully.
A mechanism is ex-post individual rational (ex-post IR) if the allocation is 
non-negative for all value profile.

\paragraph{Optimal auctions with budgets}

\citet{LR-96} and \citet{mas-00} for the revenue and welfare
objectives, respectively, characterize the Bayesian optimal mechanisms
for agents with public budgets.  With the following regularity
assumptions on the distribution, defined distinctly for revenue and
welfare, the optimal mechanism has a nice form.

\begin{definition}
A single-dimensional public budget agent is \emph{public-budget
  regular for welfare} if her cumulative distribution function
$\dist(\cdot)$ is concave; she is \emph{public-budget regular for
  revenue} if additionally $\val - \frac{1 -
  F(\val)}{f(\val)}$ is non-decreasing.
\end{definition}

The results of \citeauthor{LR-96} and \citeauthor{mas-00} can be
reinterpreted, \`a la \citet{AFHH-13}, as solving a single-agent
interim optimization problem that is given by an interim constraint
$\calloc(\cdot)$.  An interim allocation is interim feasible under the
interim constraint $\calloc(\cdot)$ if for all values $\val \in
[0,\highestval]$, the probability of allocating item to an agent with
value greater than $\val$ with allocation rule $\alloc(\cdot)$ is at
most that with allocation rule $\calloc$, i.e.\ $\int_v^\highestval
\alloc(t)\,d\dist(t) \leq \int_v^\highestval \calloc(t)\,d\dist(t)$.
In many cases solution to these interim optimization problems will
take the form of the original constraint with ironed interval and
reserve.  Ironing on arbitrary interval $[\valp,\valdp]$ corresponds
to the distribution weighted averaging as follow, $\alloc(\val) =
\int_{\valp}^{\valdp} \calloc(t)\,d\dist(t)$ for all $\val \in
    [\valp,\valdp]$.  Reserve at value $\valp$ corresponds to
    rejecting all value below $\valp$ as follows, $\alloc(\val) = 0$
    for all $\val \in [0,\valp]$.  An important allocation constraint
    is that given by the {\em highest-bid-wins} allocation rule.  The
    highest-bid-wins allocation rule for $n$ agents and with values
    from cumulative distribution function $\dist$ is $\calloc(\val) =
    \dist^{n-1}(\val)$, e.g., for two agents with uniform values it is
    $\calloc(\val) = \val$.

\begin{theorem}[\citealp{LR-96}; \citealp{mas-00}; \citealp{AFHH-13}]
\label{thm: opt welfare}\label{thm: opt revenue}
For public-budget regular i.i.d.\ agents and interim allocation
constraint $\calloc(\cdot)$, the welfare-optimal single-agent
mechanism allocates as by $\calloc(\cdot)$ except that values in
$[\valp,\highestval]$ are ironed for some $\valp$ and the
revenue-optimal single-agent mechanism additionally reserve prices
values in $[0, \valdp]$ for some $\valdp$; payments are given
deterministically by the payment identity.
\end{theorem}



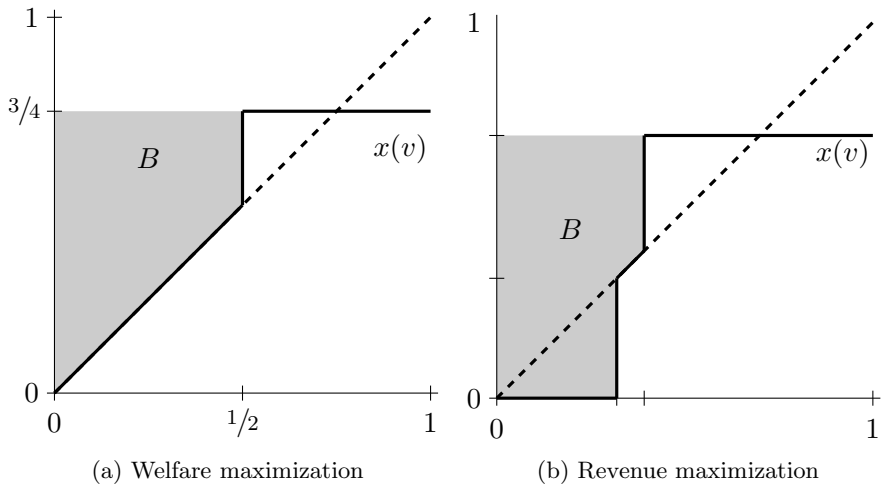
\begin{figure}
\centering
\subfloat[Welfare maximization]
{\input{all-pay-figure}}
\subfloat[Revenue maximization]
{\input{all-pay-figure2}}
\caption{Depicted are the interim allocation rules of the
  welfare-optimal and revenue-optimal mechanisms for two agents with
  uniform values on $[0,1]$.  In each figure the highest-bid-wins
  allocation rule is depicted with a dashed line.}
\end{figure}

For single-item environments, one possible implementation of
\autoref{thm: opt welfare} is the all-pay auction.  The all-pay
auction has a unique Bayes-Nash equilibrium which is identical to
outcome described in \autoref{thm: opt welfare} for the allocation
constraint given by the highest-bid-wins allocation rule.

\begin{definition}[all-pay auction]
The \emph{all-pay auction} is a mechanism $(\ballocs, \bprices)$ where $\ballocs(\cdot)$ allocates item to the agent with highest bid with tie broken at random and $\bprices(\cdot)$ charges each agent their bid,
i.e.\ $\bpricei[i](\bids) = \bidi[i]$.
\end{definition}


\begin{theorem}[\citealp{mas-00}]\label{thm: all-pay regular}
For public-budget regular i.i.d.\ agents, the all-pay auction is
welfare optimal.
\end{theorem}

%% file: all-pay-figure.tex
\begin{tikzpicture}[scale = 0.5]
\fill[gray!40!white] (0, 0) -- (5, 5) -- (5, 7.5) -- (0,7.5);

\draw (-0.2,0) -- (10.2, 0);
\draw (0, -0.2) -- (0, 10.2);
\draw (5, 0.2) -- (5, -0.2);
\draw (10, 0.2) -- (10, -0.2);
\draw (-0.2, 7.5) -- (0.2, 7.5);
\draw (-0.2, 10) -- (0.2, 10);

\begin{scope}[very thick]
\draw (0, 0) -- (5, 5);
\draw (5, 5) -- (5, 7.5);
\draw (5, 7.5) -- (10,7.5);

\draw [dashed] (0, 0) -- (10, 10);
\end{scope}

\draw (0, -0.8) node {$0$};
\draw (5, -0.8) node {$\sfrac{1}{2}$};
\draw (10, -0.8) node {$1$};

\draw (-0.6, 0) node {$0$};
\draw (-0.8, 7.5) node {$\sfrac{3}{4}$};
\draw (-0.6, 10) node {$1$};

\draw (2.5, 6.25) node {$\budget$};

\draw (9.2, 6.5) node {$\alloc(\val)$};
\end{tikzpicture}

%% file: all-pay-figure2.tex
\begin{tikzpicture}[scale = 0.5]
\fill[gray!40!white] (0, 0) -- (3.19, 0) -- (3.19, 3.19) -- (3.92, 3.92) -- (3.92, 6.99) 
-- (0, 6.99);

\draw (-0.2,0) -- (10.2, 0);
\draw (0, -0.2) -- (0, 10.2);
\draw (3.19, 0.2) -- (3.19, -0.2);
\draw (3.92, 0.2) -- (3.92, -0.2);
\draw (10, 0.2) -- (10, -0.2);
\draw (-0.2, 3.19) -- (0.2, 3.19);
\draw (-0.2, 6.99) -- (0.2, 6.99);

\begin{scope}[very thick]
\draw (0, 0) -- (3.19, 0);
\draw (3.19, 0) -- (3.19, 3.19);
\draw (3.19, 3.19) -- (3.92, 3.92);
\draw (3.92, 3.92) -- (3.92, 6.99);
\draw (3.92, 6.99) -- (10, 6.99);

\draw [dashed] (0, 0) -- (10, 10);
\end{scope}

\draw (0, -0.8) node {$0$};
\draw (10, -0.8) node {$1$};

\draw (-0.6, 0) node {$0$};
\draw (-0.6, 10) node {$1$};

\draw (1.96, 4.495) node {$\budget$};

\draw (9.2, 6.5) node {$\alloc(\val)$};
\end{tikzpicture}

%% file: clinching.tex
\section{Welfare Approximation of the Clinching Auction}
\label{s:clinching}

In this section, we study a prior-independent revelation mechanism
called the clinching auction in single-item environments.
\citet{DLN-08} gave the following formulation of the clinching auction
and characterized properties of its outcome.  See
Figure~\ref{f:clinching-auction}.

\begin{definition}[clinching auction]\label{def:clinching}
The \emph{clinching auction} maintains an allocation and price-clock starting from zero. The price-clock ascends continuously and the allocation and budget are adjusted as follows.
\begin{enumerate}
\item Agents whose values are less than price-clock are removed and their allocation is frozen.
\item The demand of any remaining agent is the remaining budget divided by the price-clock.
\item Each remaining agent clinches (and adds to their current allocation) an amount that corresponds to the largest fraction of their demand that can be satisfied when all other remaining agents are first given as much of their demand as possible.
\item The budget and allocation are updated to reflect the amount clinched in the previous step.
\end{enumerate}
\end{definition}

\begin{proposition}[\citealp{DLN-08}]
For public-budget agents, the clinching auction always allocates all
items, is ex-post IR, and is DSIC.
\end{proposition}
\begin{lemma}[a special case of \citealp{DLN-08}]
\label{lem:CL}
In single-item environment, for public-budget agents with budget $\budget$ and value profile $\vals$, and let $k$ be the largest integer such that 
$$\val_{(k)} \geq \budget\cdot k\qquad \text{
where $\val_{(k)}$ is the $k$-th highest value in value profile $\vals$.}
$$
Then, the agents with highest $(k-1)$ values win with same probability greater or equal to $\frac{1}{k}$ and the agent with the $k$-th highest value wins with the remaining probability.
\end{lemma}

We use the following approach to show that the clinching auction is an
$e$-approximation for public-budget regular agents.  We relax the
feasibility constraint to an ex ante constraint and show that the
optimal mechanism that sells to each agent with ex ante probability
$1/n$ simply posts a price (of exactly $\budget$ assuming that the
budget binds) for a chance to win the item (\Cref{lem:PO},
below).  This simple form of mechanism is closely approximated by the clinching auction which sells $k$ lotteries of $1/k$
probability (full details given subsequently).  A key property is that
in this clinching auction with lotteries, the budget does not bind with constant
probability.  The probability that the budget does not bind in the clinching auction with lotteries allows a
lower bound on the allocation probability in the clinching auction
which allows its welfare to be compared to the ex ante relaxation.

Consider the welfare-optimal auction. Since agents are symmetric, each
agent will win with ex ante probability exactly $\frac{1}{n}$ in the
welfare-optimal auction where $n$ is the number of agents.  We replace
the feasibility constraint that ex post allocation cannot allocate
more than one item (i.e. $\sum_{i\in[n]}\balloci(\vals) \leq 1$ for
all $\vals$) with a $\frac{1}{n}$ ex ante constraint that each agent cannot be
allocated more than $\frac{1}{n}$ in expectation
(i.e.\ $\expect[\val]{\alloc(\val)} \leq \sfrac{1}{n}$).  Ex ante optimal
mechanisms for agents with public budgets were proposed and studied by
\citet{AFHH-13}. 

\begin{figure}[t]
\centering
\input{clinching-figure}
\caption{\label{f:clinching-vs-ex-ante} The allocation rules of the ex
  ante relaxation (dashed), an $1/e$-fraction of the ex ante
  relaxation (dotted), and the clinching auction with lotteries (solid) are
  depicted. The clinching auction with lotteries pointwise exceeds an $1/e$-fraction
  of the ex ante relaxation.}
\end{figure}
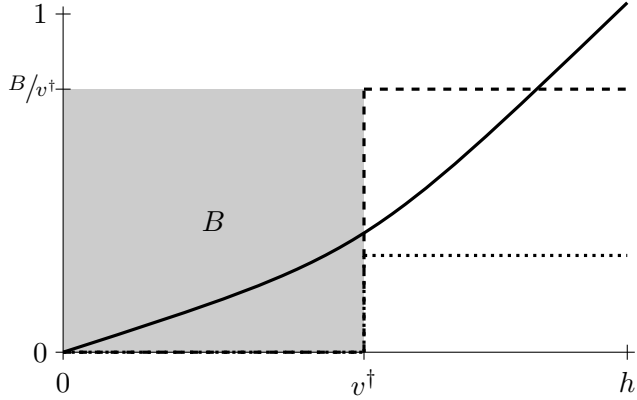

\begin{lemma}[\citealp{AFHH-13}]\label{lem:PO}
For public-budget regular i.i.d.\ agents with budget $\budget$, the
{\em ex ante welfare-optimal mechanism} is either:
\begin{enumerate}[i.]
\item\label{budget binds} Budget binds: 
Post the price $\budget$ for allocation probability $\frac{\budget}{  \val \primed} \leq 1$ with 
$\val \primed$ set to satisfy 
$\frac{1}{n} = \frac{\budget}{  \val\primed}(1 - F(  \val \primed))$. 
Values $\val \geq \val \primed$ select the lottery.

\item\label{allocation binds} Allocation probability binds: Post price $  \val\primed = F^{-1}(1 - \frac{1}{n})$
for allocation probability one.
\end{enumerate}
\end{lemma}

We build the connection between the clinching auction and the ex ante
optimal mechanism by considering the an additional auction: the
clinching auction with lotteries $\CL$ which allocates $k$ lotteries
with winning probability $1/k$ per lottery, using the clinching
auction framework under the same public budget.  \Cref{lem:PO-CL}
below shows that by selecting an appropriate $k$, the probability that
an agent with value $\valp$ wins in the clinching auction with
lotteries $\CL$ is at least an $e$ fraction of the probability that
the agent (with value $\valp$) wins in the ex ante relaxation. See
\Cref{f:clinching-vs-ex-ante}.

\begin{lemma}\label{lem:PO-CL}
For public budget i.i.d.\ agents, at value $\valp$ defined in
\Cref{lem:PO}, 
there exists $k\in[n]$, 
such that the interim allocation of the clinching auction with lotteries $\allocCL(\valp)$ is
an $e$-approximation of the interim allocation of the ex ante optimal
mechanism $\allocPO(\valp)$, i.e.,
$
\allocCL(\valp) \geq \sfrac{1}{e}\cdot\allocPO(\valp).
$
\end{lemma}

\begin{proof}
Denote the
notation $\val_{(j:m)}$ as the $j$-th order statistic among $m$ i.i.d.\ random variables from distribution $F$.

We denote the posted pricing in \Cref{lem:PO} as $\PO$. Let $k_0 =
1/\allocPO(\highestval)$ where $\allocPO(\highestval)$ is the interim
allocation at the highest value $\highestval$.  By the construction of
$\PO$, $F(\valp) = 1 - \tfrac{k_0}{n}$ and $\valp \leq \budget \cdot
k_0$ (equality holds when the budget binds in $\PO$).  Let $k = \lceil
k_0\rceil$ be the smallest integer which is greater or equal to $k_0$.
Consider the clinching auction $\CL$ which allocates $k$ lotteries
with winning probability $1/k$ per lottery, using the clinching
auction framework under public budget $\budget$.

First, fix an arbitrary agent and fix her value to be $\valp$, we
consider the following event $\event$: in $\CL$, this agent with value
$\valp$ is one of the highest $k$ valued agents and the budget does
not bind.  Recall that when the budget does not bind, the highest $k$
agents in $\CL$ each receive lotteries (with allocation probability
$1/k$) and pay the value of the $(k+1)$-st highest agent divided by
$k$ (i.e.\ $\val_{(k+1:n)}/k$).  The budget bids in $\CL$ if and only
if $\val_{(k+1:n)}/k \leq \budget$ and we can lower bound the lower
bound the probability of the event $\event$ as follows,
\begin{align*}
\prob{\event} 
&=
\prob{(\val_{(k:n-1)} / k \leq B) \land (\val_{(k:n-1)} \leq \valp)}\\
&=
\prob{(v_{(k:n-1)} / k \leq \valp / k_0) \land (\val_{(k:n-1)} \leq \valp)} \\
&=
\prob{\val_{(k:n-1)} \leq \valp} \\
&=
\sum_{i = 0}^{k - 1} \tbinom{n - 1}{i}\left(\tfrac{k_0}{n}\right)^i\left(\tfrac{n - k_0}{n}\right)^{n - 1 - i}.\\
\intertext{
Above, the third line is derived from the second line using the definition of $k \geq k_0$.  Denote by $\alloc^{\event}$ and $\alloc^{\bar\event}$ the allocation
rule $\alloc$ conditioned on the events $\event$ and $\bar\event$,
respectively.
The interim allocation for $\CL$ at value $\valp$ can be lower bounded as follows.}
\allocCL (\valp) 
&= 
\allocCL^{\event}(\valp)\cdot \prob{\event} + 
\allocCL^{\bar \event}(\valp)\cdot \prob{\bar\event}
\\
&\geq
\allocCL^{\event}(\valp) \cdot \prob{\event} 
\\
&=
\tfrac{k_0}{k}\cdot \allocPO(\highestval) \cdot \prob{\event}\\
&\geq
\tfrac{1}{e}\cdot\allocPO(\highestval).
\end{align*}
The final inequality follows because the term $\tfrac{k_0}{k}\cdot
\prob{\event}$ achieves the minimum at $1/e$ when $k_0 = k = 1$ and $n$ goes to
infinity.  
\end{proof}

We now prove our main theorem about the approximation ratio for the
clinching auction.

\begin{theorem}\label{thm:clinching auction e approx}
For public-budget regular i.i.d.\ agents, the clinching auction is an $e$-approximation to the welfare-optimal mechanism.
\end{theorem}

\begin{proof}
By \Cref{lem:PO} the interim allocation rule of the ex ante optimal
mechanism is a step function that steps at value $\valp$.  By
\Cref{lem:PO-CL}, at value $\valp$, the allocation rule of the
clinching auction with lotteries is an $e$-approximation to that of
the ex ante optimal mechanism.  The allocation rule of the clinching
auction with lotteries is monotone, so its allocation rule is an
$e$-approximation to that of the ex ante optimal mechanism at every
value.  Consequently, the expected welfare of the clinching auction
with lotteries is at least an $e$-approximation to that of the ex ante
relaxation.  See \autoref{f:clinching-vs-ex-ante}.

Finally, \Cref{lem:CL} implies that for every ex post value profile,
the  welfare of
the clinching auction is at least that of the clinching
auction with lotteries.
\end{proof}

For public-budget regular i.i.d.\ agents, the all-pay auction is optimal while the clinching auction is not, since the budget binds for more value profiles in the clinching auction than in the all-pay auction. 
Based on this, we give a $\clinchinglowerbound$ lower bound of the approximation ratio for the clinching auction and leave the actual approximation ratio as an open problem.
\begin{lemma}
There exists the instance of public-budget regular agents where the clinching auction is a $\clinchinglowerbound$-approximation of the welfare-optimal mechanism.
\end{lemma}

\begin{proof}
Consider a simple setting: there are 2 public-budget regular agents with value drawn 
uniformly from $[0, \highestval]$ and the budget $\budget = 1$.

By \Cref{thm: all-pay regular}, the all-pay auction is welfare-optimal for
public-budget regular agents.
The interim allocation rule of it is 
$\alloc(\val) = \frac{\val}{\highestval}$ if $\val \leq 2$ 
and 
$\alloc(\val) = \frac{\highestval+2}{2\highestval}$ otherwise. 
The expected welfare of all-pay auction is $(3\highestval^3+6\highestval^2-12\highestval+8)/6\highestval^2$.  

The interim allocation rule of the clinching auction is 
$\alloc(\val) = \frac{\val}{\highestval}$ if $\val \leq 1$ and 
$\alloc(\val) = \frac{\highestval+2}{2\highestval} - \frac{1}{2\val^2}$ otherwise. 
The expected welfare of the clinching auction is 
$(3\highestval^3+6\highestval^2-3\highestval-6\highestval\ln \highestval - 2)/6\highestval^2$. 

By setting $\highestval = 4.04$, 
it optimizes the ratio as $\clinchinglowerbound$. 
\end{proof}

%% file: clinching-figure.tex
\begin{tikzpicture}[scale = 0.5]
\fill[gray!40!white] (0, 0) -- (8, 0) -- (8, 7) -- (0, 7);

\draw (4, 3.5) node {$\budget$};

\draw (-0.2,0) -- (15.2, 0);
\draw (0, -0.2) -- (0, 9.2);
\draw (8, 0.2) -- (8, -0.2);
\draw (15, 0.2) -- (15, -0.2);
\draw (-0.2, 7) -- (0.2, 7);
\draw (-0.2, 9) -- (0.2, 9);

\begin{scope}[very thick]
\draw [dashed] (0, 0) -- (8, 0);
\draw [dashed] (8, 0) -- (8, 7);
\draw [dashed] (8, 7) -- (15,7);

\draw [dotted] (0, 0) -- (8, 0);
\draw [dotted] (8, 0) -- (8, 2.575);
\draw [dotted] (8, 2.575) -- (15, 2.575); 

\draw (0, 0) .. controls (8, 2.575) .. (15,9.3); 
\end{scope}

\draw (0, -0.8) node {$0$};
\draw (8, -0.8) node {$\valp$};
\draw (15, -0.8) node {$\highestval$};

\draw (-0.6, 0) node {$0$};
\draw (-0.8, 7) node {$\sfrac{\budget}{\valp}$};
\draw (-0.6, 9) node {$1$};

\end{tikzpicture}

%% file: optimal_DSIC.tex
\section{Bayesian Optimal DSIC Mechanism}
\label{s:DSIC}
In \Cref{thm: all-pay regular}, the all-pay auction is
welfare-optimal under public-budget regular distribution.  Hence,
applying the revelation principle to the all-pay auction, it produces
a Bayesian optimal revelation mechanism.  This mechanism is BIC but
not DSIC.  In this section, we characterize the optimal DSIC mechanism
for two agents with uniformly distributed values.  We obtain a
lower bound on its approximation ratio with the BIC optimal mechanism.

\begin{figure}
\centering
\subfloat
[The middle-ironed clinching auction]
{
\label{f:middle-ironed-clinching-auction}
\input{DSIC-figure-alloc}
}
\subfloat
[The clinching auction]
{
\label{f:clinching-auction}
\input{DSIC-figure-alloc2}
}
\caption{ The comparison of the allocation rule
  $\alloc_1(\val_1,\val_2)$ for the middle-ironed clinching auction
  and the clinching auction.  In the middle-ironed clinching auction,
  for the values in interval $\mvalue$ can be thought as ``ironed'',
  i.e.\ an agent receives the same outcome for any value $\val \in \mvalue$.}
\end{figure}
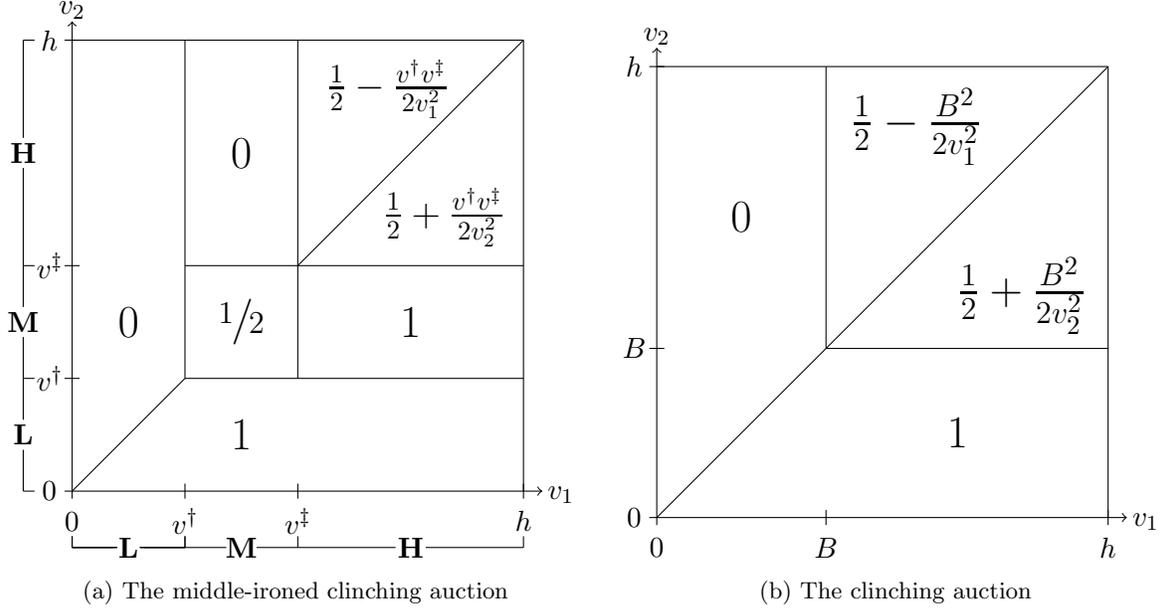

We first introduce the middle-ironed clinching auction (for two agents). 
\begin{definition}
\label{d:middle-ironed-clinching-auction}
The two-agent {\em middle-ironed clinching auction} is parameterized
by $\valp \leq \budget$ and $\valdp = 2\budget - \valp$ and its
outcome is highest-bid-wins on values less that $\valp$, a fair
lottery on values in $[\valp,\valdp]$, and the clinching auction on
values exceeding $\valdp$; a precise formulation for two-agents is
given in Figure~\ref{f:middle-ironed-clinching-auction} and a general
formulation is given in \autoref{app:price-jump}.
\end{definition}

For two-agents case, the middle-ironed clinching auction allocates the
item efficiently except for value profiles in $\mvalue\mvalue$ (both
agents with values in $\mvalue$) or $\hvalue\hvalue$ (both agents with
values in $\hvalue$).  For the value profile in $\mvalue\mvalue$, it
randomly allocates the item to one of the agent with probability
$\tfrac{1}{2}$ with payment $\tfrac{\valp}{2}$.  For the value profile
in $\hvalue\hvalue$, it allocates the item such that the budget binds
for the agent with higher value and the allocation rule depends on the
lower value only.  Figure~\ref{f:clinching-auction} depicts the
allocation rule of the clinching auction for comparison.  The
middle-ironed clinching auction can be implemented with an ascending
price via a generalization of the clinching auction that allows for
price jumps which we develop in Appendix~\ref{app:price-jump} (this
generalization is non-trivial).

We will show that by selecting the proper thresholds $\valp$ and
$\valdp$, the middle-ironed clinching auction is the Bayesian optimal
DSIC mechanism for two agents with uniformly distributed values.  An
intuition behind the optimality of the middle-ironed clinching auction
is as follows: \citet{DLN-08} show that for two public budget agents,
the clinching auction is the only Pareto optimal (i.e.\ there is no
outcome which is weakly better for all agents and strictly better for
one agent) and DSIC auction.  Moreover, after the price increases past
the point where the budget binds, a differential equation governs the
allocation of any DSIC mechanism.  Our goal is to optimize expected
welfare rather than satisfy Pareto optimality.  Sacrificing welfare
for lower-valued agents by ironing can delay the budget from binding
and enable greater welfare from higher-valued agents.  From our proof
of optimality, it is sufficient to only iron one region in the middle
of value space.

\begin{theorem}\label{thm: middle-ironed clinching}
For two public-budget agents with budget $\budget$ and value uniformly
drawn from $[0,\highestval]$, Bayesian optimal DSIC mechanism is the
middle-ironed clinching auction with some thresholds $\valp$ and
$\valdp$.
\end{theorem}

The approach of the proof is to write down our
problem as a linear program (primal), assume the middle-ironed
clinching auction to be the solution, and then construct the dual
program with a dual solution which witnesses the optimality of the
primal solution by complementary slackness.  This approach is
reminiscent of that of \citet{PV-14} and \citet{DW-17}; however, our
multi-agent DSIC constrained program presents novel challenges and for
this reason we only solve the problem of two agents and uniform
distributions.

We first solve a discrete version of the problem. 
Then, we solve the continuous version as the limit from the discrete version. 
Consider the value distribution with finite value space $[\highestval] = \{1, 2,\dots, \highestval\}$ with equal probability each. 
We begin by writing down the optimization program for welfare maximization among all possible DSIC mechanism.

$$
\begin{array}{ll}
\sup\limits_{(\allocs,\prices)} 
&\sum\limits_{\val_1, \val_2 \in [\highestval]}
(\val_1\cdot\alloc_1(\val_1,\val_2)
+
\val_2\cdot\alloc_2(\val_1,\val_2))
\cdot
\tfrac{1}{\highestval}
\cdot
\tfrac{1}{\highestval}
\\
s.t.\\&\ 
\begin{array}{ll}
(\allocs,\prices) \text{ are DSIC, ex-post IR},
\text{and
feasible}
\\
(\allocs, \prices) \text{ is budget balanced} 
\end{array}
\end{array}
$$


By the characterization of dominant strategy equilibrium, we simplify this optimization program into a linear program as follows,
$$
\begin{array}{ll}
\max\limits_{(\allocs,\prices)\geq 0} 
&\sum\limits_{\val_1, \val_2 \in [\highestval]}
\val_1\cdot\alloc(\val_1,\val_2)
\\
s.t.\\&\ 
\begin{array}{lll}
\highestval \cdot \alloc(\highestval, \val_2) - \sum_{t = 1}^\highestval \alloc(t,\val_2) \leq \budget
\qquad &\text{for all }\val_2 \in [\highestval] 
\qquad &\text{[Budget Constraint]}
\\
\alloc(\val_1,\val_2) + \alloc(\val_2, \val_1) \leq 1
&\text{for all } \val_1,\val_2 \in [\highestval]
&\text{[Feasibility Constraint]}
\\
\alloc(\val_1,\val_2) \leq \alloc(\val_1 + 1,\val_2)
&\text{for all } \val_1 \in [\highestval - 1], \val_2 \in [\highestval]
&\text{[Monotonicity Constraint]}
\end{array}
\end{array}
$$
where we assume $\alloc_1(a, b) = \alloc_2(b, a) = \alloc(a, b)$ for all $a, b \in [\highestval]$ since it is an agent-symmetric linear program.
\footnote{
Note the program in turns of $\alloc(a, b)$ is asymmetric.
}

Additionally, we relax the monotonicity constraint by replacing it with $\alloc(\val_1, \val_2) \leq \alloc(\highestval, \val_2)$
which is common for Bayesian mechanism design with public budget agents.
$$
\begin{array}{lll}
\alloc(\val_1,\val_2) \leq \alloc(\highestval,\val_2)
&\text{for all } \val_1 \in [\highestval - 1], \val_2 \in [\highestval]
&\text{[Relaxed Monotonicity Constraint]}
\end{array}
$$ 

Then we write down the corresponding dual program.
Let $\{\Lambda(\val_2)\}_{\val_2\in[h]}$ denote the dual variables for budget constraints; $\{\beta(\val_1, \val_2)\}_{\val_1,\val_2\in[\highestval]}$ denote the dual variables for feasibility constraints (for simplicity, we use both $\beta(\val_1, \val_2)$ and $\beta(\val_2,\val_1)$ to denote the same dual variable); and $\{\mu(\val_1, \val_2)\}_{\val_1\in[\highestval-1],\val_2\in[\highestval]}$ denote the dual variables for monotonicity constraints. The dual program is,
$$
\begin{array}{ll}
\min\limits_{(\Lambdas,\betas,\mus)\geq 0} 
&\sum\limits_{\val_2\in [\highestval]}
\budget\cdot\Lambda(\val_2)
+
\tfrac{1}{2}\sum\limits_{\val_1,\val_2\in [\highestval]}\beta(\val_1,\val_2)
\\
s.t.\\&\ 
\begin{array}{lll}
-\Lambda(\val_2) + \beta(\val_1,\val_2) + \mu(\val_1, \val_2) \geq \val_1
\qquad &\text{for all }\val_1 \in [\highestval - 1], \val_2 \in [\highestval] 
\qquad &\text{[$\alloc(\val_1, \val_2)$]}
\\
(\highestval - 1)\Lambda(\val_2) + \beta(\highestval,\val_2) - \sum_{t = 1}^{\highestval - 1}\mu(t,\val_2) \geq \highestval
\qquad &\text{for all }\val_2 \in [\highestval] 
\qquad &\text{[$\alloc(\highestval, \val_2)$]}
\end{array}
\end{array}
$$

We give a short overview of the plan to solve the program.
For each possible thresholds $\valp, \valdp$ chosen in the middle-ironed clinching auction, 
we first construct a solution in dual which satisfies the complementary slackness
with this middle-ironed clinching auction as a solution in primal.
These induced dual solutions may be infeasible.
Next, we will show that there exists a pair of thresholds $\valp, \valdp$ which induces a feasible dual solution. 
This feasible dual solution witnesses the optimality of the middle-ironed clinching auction. 

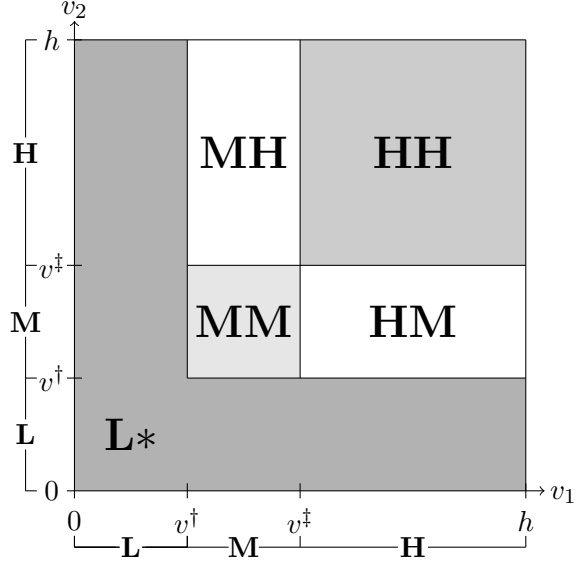
\begin{figure}
\centering
\input{DSIC-figure-partition}
\caption{
We partition the dual variables into 
$\LS$ (at least one agent with value in $\lvalue$),
$\HH$ (both agents with values in $\hvalue$),
$\MM$ (both agents with values in $\mvalue$),
$\MH$ and $\HM$ (one agent with value in $\mvalue$ and the other with value in $\hvalue$)
five areas.
}
\label{fig: DSIC-partition}
\end{figure}

We will partition the dual variables into following five areas ($\LS$, $\MM$, $\HH$, $\MH$ and $\HM$)
as in Figure \ref{fig: DSIC-partition};
and construct the dual solution for them separately.
We denote $\lambdas$ as the discrete derivative of the dual variable $\Lambdas$, 
i.e.\ $\lambda(\val) = \Lambda(\val) - \Lambda(\val + 1)$.
\begin{description}
\item [$\Lambdas$ in $\lvalue$:] 
Since the budget constraints do not bind,
by complementary slackness,
$$\Lambda(\val) = 0\text{ for all }\val \in \lvalue.$$

\item[$\betas, \mus$ in $\LS$:] 
Let $(\val, \val')$ be a value profile in area $\LS$ such that $\val \geq \val'$.
By complementary slackness on $\alloc(\val, \val')$,
$\beta(\val, \val') + \mu(\val, \val') - \Lambda(\val') = \val$
if $\val < \highestval$; 
$\beta(\val, \val') - \sum_{t=1}^{\highestval - 1}\mu(t, \val') + (\highestval - 1)\Lambda(\val') = \val$ 
otherwise (i.e.\ $\val = \highestval)$.
We let
$$\beta(\val, \val') = \val\text{ and }\mu(\val, \val') = 0.
\footnote{An intuition here is: $\mus$ are the dual variables for the relaxed monotonicity constraint and can be thought as indicators of ironing. Though the monotonicity constraint binds, this is not because of ironing but binding allocation (i.e. $\alloc(\cdot) \leq 1$). Therefore, we set $\mus$ as zero.}
$$

Since the relaxed monotonicity constraint does not bind at $\alloc(\val',\val)$,
i.e. $\alloc(\val', \val) < \alloc(\highestval, \val)$, 
the corresponding dual variable is 
$$\mu(\val', \val) = 0.$$

\item[$\betas, \mus$ in $\HH$:]
Let $(\val, \val')$ be a value profile in area $\HH$ such that $\val \geq \val'$.
Since both agents win with non-zero probability,
by complementary slackness on $\alloc(\val, \val')$ and $\alloc(\val',\val)$, the corresponding dual constraints bind.
Since the relaxed monotonicity constraint does not bind at $\alloc(\val',\val)$,
the monotonicity dual variable is 
$$\mu(\val', \val) = 0.$$
The binding dual constraint of $\alloc(\val',\val)$ is 
$\beta(\val',\val) - \Lambda(\val) + \mu(\val',\val) = \val'$. Hence, 
$$\beta(\val, \val') = \beta(\val', \val) = \val' + \Lambda(\val).
\footnote{Recall that $\beta(\val,\val')$ and $\beta(\val',\val)$ denote the same dual variable.}
$$
The binding dual constraint of $\alloc(\val, \val')$ is 
$\beta(\val,\val') - \Lambda(\val') + \mu(\val,\val') = \val$.
Note the relaxed monotonicity constraint is tight for $(\val, \val')$.
Hence, 
$$\mu(\val,\val') 
= 
\val - \val' + \Lambda(\val') - \Lambda(\val).
$$
Here we write $\betas, \mus$ as terms of $\Lambdas$. In the next paragraph, we will solve for $\Lambdas$.

\item[$\Lambdas$ in $\hvalue$:]
Let $\val \in \hvalue$.
Consider the binding dual constraint of $\alloc(\highestval, \val)$,
$(\highestval - 1)\Lambda(\val) + \beta(\highestval, \val) - \sum_{t = 1}^{\highestval - 1}\mu(t,\val) = \highestval$.
Notice that by complementary slackness, $\mu(t,\val) = 0$ for all $t \leq \val$.
Plugging $\betas$ and $\mus$ as terms of $\Lambdas$ into the these dual constraints of $\alloc(\highestval, \val)$, 
we can solve for $\Lambdas$ as
$$\lambda(\val) = \frac{\highestval - \val}{\val}
\text{ for all }v \in \hvalue
\text{ and }
\Lambda(\highestval) = 0.
\footnote{Recall that $\lambdas$ is the discrete derivative of dual variables $\Lambdas$, so $\Lambda(\val) = \sum_{t = \val}^{\highestval - 1}\lambda(\val)$.}
$$

\item[$\betas, \mus$ in $\MM$
 and $\Lambdas$ in $\mvalue$:] 
Let $(\val, \val')$ be a value profile in area $\MM$ such that $\val \geq \val'$.
Since the relaxed monotonicity constraints do not bind for either $\alloc(\val, \val')$ or
$\alloc(\val', \val)$,
the corresponding dual variables are
$$\mu(\val,\val') = \mu(\val',\val) = 0.$$
The binding dual constraints of $\alloc(\val,\val')$ implies
$\beta(\val, \val') = \val' + \Lambda(\val)$.
On the other hand, the binding dual constraints of $\alloc(\val',\val)$ implies
$\beta(\val',\val) = \val + \Lambda(\val')$. 
Recall that $\beta(\val,\val')$ and $\beta(\val',\val)$ denote the same variable, hence,
$$
\lambda(\val) = -1 \text{ for all } \val \in \mvalue \setminus \{\valdp - 1\},
\footnote{
Complementary slackness does not pin down $\dvp$. 
We leave it as a variable 
and identify it later when we choose the thresholds 
$\valp$, $\valdp$ to 
ensure that the dual solution is feasible.
}
$$
$$
\beta(\val,\val') = 
\Lambda(\valdp) + \dvp + \val + \val' - \valdp.
$$

\item[$\betas,\mus$ in $\MH$ and $\HM$:] 
Let $(\val, \val')$ be a value profile in area $\HM$ such that $\val > \val'$.
With the similar argument for area $\HH$,
$$\mu(\val',\val) = 0 \text{ and } \mu(\val,\val') = \val - \val' + \Lambda(\val') - \Lambda(\val),$$ 
$$\beta(\val,\val') = \val' + \Lambda(\val)\text{ if }\val < \highestval.$$
Plugging $\mus$ as terms of $\Lambdas$ into the binding dual constraint of $\alloc(\highestval, \val')$,
$$\beta(\highestval, \val') = 
(\highestval - 1)(\valdp - \val') + 1 + (\valdp - 1)\dvp.$$
\end{description}

With the analysis above, we construct the following dual solution which satisfies complementary slackness with the middle-ironed clinching auction as a solution in primal,
\begin{align}\label{dual construction}
\begin{split}
\Lambda(\val_2) &= 
\left\{
\begin{array}{ll}
0 \qquad &\text{if }\val_2 < \valp\\
\sum_{k=\valdp}^{\highestval - 1} \frac{\highestval - k}{k} + \val_2 - \valdp + 1  + \dvp
&\text{if }\valp \leq \val_2 < \valdp\\
\sum_{k=\val_2}^{\highestval - 1} \frac{\highestval - k}{k} 
&\text{if }\val_2 \geq \valdp\\
\end{array}
\right.
\\
\beta(\val_1,\val_2) &=
\left\{
\begin{array}{ll}
\val_1 
\qquad &\text{if }\val_1 \geq \val_2,\ \val_2 < \valp \\
\sum_{k=\valdp}^{\highestval-1}\frac{\highestval - k}{k} + \val_1 + \val_2 -\valdp + 1 + \dvp
\qquad &\text{if }\val_1 \geq \val_2,\ \valp \leq \val_2 < \valdp,\ \val_1 < \valdp \\
\sum_{k=\val_1}^{\highestval-1}\frac{\highestval - k}{k} + \val_2
\qquad &\text{if }\val_1 \geq \val_2,\ \valp \leq \val_2 < \valdp,\ \valdp \leq \val_1 \leq \highestval - 1\\
(\highestval - 1)(\valdp - \val_2) + 1 + (\valdp - 1)\dvp
\qquad &\text{if }\val_1 \geq \val_2,\ \valp \leq \val_2 < \valdp,\ \val_1 = \highestval\\
\sum_{k=\val_1}^{\highestval-1}\frac{\highestval - k}{k} + \val_2
\qquad &\text{if }\val_1 \geq \val_2,\ \valdp \leq \val_2
\end{array}
\right.
\\
\mu(\val_1,\val_2) &=
\left\{
\begin{array}{ll}
0 \qquad &\text{if }\val_2 < \valp\\
0 
&\text{if } \valp \leq \val_2 < \valdp, \val_1 < \valdp \\
\val_1 - \valdp + \sum_{k=\valdp}^{\val_1-1}\frac{\highestval - k}{k} + 1 + \dvp
&\text{if } \valp \leq \val_2 < \valdp, \val_1 \geq \valdp \\
0 
&\text{if }  \val_2 \geq \valdp, \val_1 \leq \val_2 \\
\val_1 - \val_2 + \sum_{k=\val_2}^{\val_1-1}\frac{\highestval - k}{k}
&\text{if }  \val_2 \geq \valdp, \val_1 > \val_2 \\
\end{array}
\right.
\end{split}
\end{align}

\begin{lemma}
For the middle-ironed clinching auction with arbitrary thresholds $\valp$ and $\valdp$,
the dual solution 
\eqref{dual construction} satisfies the complementary slackness.
\end{lemma}
\begin{proof}
The complementary slackness is directly implied by the construction above.
\end{proof}

Though the this dual solution satisfies the complementary slackness, it may be infeasible. 
Therefore, we argue that there exists some thresholds $\valp$, $\valdp$ and $\dvp$ under which the dual solution is feasible.

\begin{lemma}
There exists $\valp, \valdp$ and $\dvp$ such that the constructed dual solution 
\eqref{dual construction} is feasible.
\end{lemma}
\begin{proof}
We define function $\mt(v) = 2v - 2\budget - 2 - \sum_{k = v}^{\highestval - 1}\frac{\highestval - k}{k}$ 
to simplify the argument.
Notice that
$\Lambda(v) = \sum_{k = v}^{\highestval-1}\frac{\highestval - k}{k}$ 
in the dual solution \eqref{dual construction} if $v \in \hvalue$.

Due to complementary slackness, all dual constraints corresponding to some $\alloc(\val, \val') > 0$ 
bind, so they are satisfied automatically.
Hence, to ensure the constructed dual solution is feasible, 
there remain four groups of constraints which need to be satisfied.
For each group of constraints, there is a ``pivotal'' constraint such that 
if it is satisfied, all constraints in that group is satisfied.
We list these four groups of constraints and ``pivotal'' constraint for each group below,
\begin{description}
\item 
[All dual constraints of $\alloc(\val,\val')$ where $\val \in \lvalue$ and $\val' \in \mvalue$:]
The pivotal constraint is the dual constraint of $\alloc(\valp - 1, 1)$, which can be simplified as
$$\dvp \leq \mt(\valdp).$$

\item
[All dual constraints of $\alloc(\val,\val')$ where $\val \in \lvalue$ and $\val' \in \hvalue$:]
The pivotal constraint is the dual constraint of $\alloc(\valp - 1, \valdp)$, which can be simplified as
$$-1 \leq \mt(\valdp).$$

\item
[All dual constraints of $\alloc(\val, \highestval)$ where $\val \in \mvalue$:]
The pivotal constraint is the dual constraint of $\alloc(\valdp - 1, \highestval)$, which can be simplified as
$$\dvp \leq \frac{\highestval}{\valdp - 1} - 1.$$

\item [$\Lambdas,\mus,\betas \geq 0$:]
The pivotal constraint is $\Lambda(\valp) \geq 0$, which can be simplified as
$$\dvp \geq \mt(\valdp) - 1.$$
\end{description}

We now show how to relate $\valp$, $\valdp$ and $\dvp$ to satisfy the four inequalities
identified above.

Notice that when $\valp = 1$ and $\valdp = 2\budget + 1$, the interval $\lvalue$ becomes empty. 
In that case, the first and second groups of constraints disappear. The combination of these four inequalities is equivalent to 
\begin{enumerate}
\item[i.] $\valdp = 2\budget + 1$ and 
$\mt(\valdp) \leq \frac{\highestval}{\valdp - 1}$; 
or
\item[ii.] $-1 \leq \mt(\valdp) \leq \frac{\highestval}{\valdp - 1}$.
\end{enumerate}

Without loss of generality, we assume that Condition (i) does not hold 
and then argue that Condition (ii) holds in this case.

The construction of $\mt(\cdot)$ implies the following two facts, 
\begin{enumerate}
\item[(a)] if $\mt(v) < -1$, then $\mt(v + 1) < \frac{\highestval}{v}$;  
\item[(b)] if $\mt(v) > \frac{\highestval}{v - 1}$, then $\mt(v - 1) > -1$.
\end{enumerate}
If we think of interval $(-1,\frac{\highestval}{\val - 1})$ as a ``window'', 
these two facts say that 
if a point $\mt(\val)$ is on the left hand side of this window
(i.e.\ $\mt(\val) < -1$)
, then the next point $\mt(\val + 1)$ cannot jump to the right hand side of the window 
(i.e.\ it is either in the window or still on the left hand side of the window)
and vice versa.
Notice that 
$\mt(\budget + 1) = -\sum_{k = \budget + 1}^{\highestval-1}\frac{\highestval - k}{k} \leq 0 < \frac{\highestval}{\budget}$ 
and our assumption
(i.e.\ Condition (i) does not hold)
implies
$\mt(2\budget + 1) > \frac{\highestval}{2\budget} > -1$. 
Thus, there exists $\valdp$ such that $-1\leq \mt(\valdp) \leq \frac{\highestval}{\valdp - 1}$, i.e.\ Condition (ii) holds.
\qedhere
\end{proof}

The construction of the dual solution which satisfies feasibility and complementary slackness witnesses the optimality of the middle-ironed clinching auction. 
We offer the following discrete version of 
\Cref{thm: middle-ironed clinching}.

\begin{theorem}\label{thm:discrete characterization}
For two public-budget agents with value uniformly distributed from $\{1,\dots, \highestval\}$, 
the Bayesian optimal DSIC mechanism is the middle-ironed clinching auction for some thresholds $\valp$ and $\valdp$.
\end{theorem}

We now focus on continuous uniform distribution with value space $[0, \highestval]$. 
Again, we write the problem as an optimization program as follows,
$$
\begin{array}{ll}
\max\limits_{(\alloc,\price)\geq 0} 
&
\int\limits_0^\highestval
\int\limits_0^\highestval
\val_1\cdot\alloc(\val_1,\val_2)d\val_2 d\val_1
\\
s.t.\\&\ 
\begin{array}{lll}
\highestval \cdot \alloc(\highestval, \val_2) - 
\int_0^\highestval \alloc(t,\val_2)dt \leq \budget
\qquad &\text{for all }\val_2 \in [0, \highestval] 
\qquad &\text{[Budget Constraint]}
\\
\alloc(\val_1,\val_2) + \alloc(\val_2, \val_1) \leq 1
&\text{for all } \val_1,\val_2 \in [0, \highestval]
&\text{[Feasibility Constraint]}
\\
\alloc(\val_1,\val_2) \leq \alloc(\highestval,\val_2)
&\text{for all } \val_1,\val_2 \in [0, \highestval]
&\text{[Relaxed Monotonicity Constraint]}
\end{array}
\end{array}
$$ 

\begin{proof}[Proof of \Cref{thm: middle-ironed clinching}]
Discretize the value space $[0, \highestval]$ into $\{\epsilon, 2\epsilon, \dots, m\epsilon\}$ where $m\epsilon = \highestval$ with density $\tfrac{1}{m}$ each.
Define $\allocspace_\epsilon$ to be the class of all possible DSIC, ex-post IR, budget balanced allocations such that each value $\val \in [(k-1)\epsilon, k\epsilon)$ must be ironed for all $k = 1, \dots, \highestval$.
By the construction of $\allocspace_\epsilon$, the allocation function $\alloc^\epsilon$ in \Cref{thm:discrete characterization} indeed solves 
$\max_{x\in \allocspace_\epsilon}
\int_0^\highestval
\int_0^\highestval
\val_1\cdot\alloc(\val_1,\val_2)d\val_2 d\val_1
$ 
after rescaling both value space and budget by $\tfrac{1}{\epsilon}$.

Let $\allocspace$ be the class of all possible DSIC, ex-post IR, budget balanced allocations. 
Notice that $\allocspace_\epsilon$ converges to $\allocspace$ pointwise and that both are compact subsets of the $L_1$ space defined by uniform measure. 
The operator $T(\alloc) = \int_0^\highestval
\int_0^\highestval
\val_1\cdot\alloc(\val_1,\val_2)d\val_2 d\val_1$
is a bounded linear operator from the $L_1$ space of allocation function to $\mathbb R$.
Therefore, $T$ achieves its maximum on each set $\allocspace^\epsilon$ and $\allocspace$. 

The pointwise convergence ensures that 
\begin{align*}
\begin{split}
\lim\limits_{\epsilon\rightarrow 0}
\max\limits_{\alloc \in \allocspace^{\epsilon}}
\int_0^\highestval
\int_0^\highestval
\val_1\cdot\alloc(\val_1,\val_2)d\val_2 d\val_1 
=
\max\limits_{\alloc \in \allocspace}
\int_0^\highestval
\int_0^\highestval
\val_1\cdot\alloc(\val_1,\val_2)d\val_2 d\val_1 
\end{split}
\end{align*}

Since $T(\alloc)$ is a bounded linear operator and $\{\alloc^\epsilon\}$ has a pointwise limit,
$$\lim\limits_{\epsilon\rightarrow 0}\alloc^\epsilon \in \{\argmax_{\alloc\in\allocspace}
\int_0^\highestval
\int_0^\highestval
\val_1\cdot\alloc(\val_1,\val_2)d\val_2 d\val_1 \}
.
$$

Thus, we see that \Cref{thm: middle-ironed clinching} holds.
\qedhere
\end{proof}

Based on \Cref{thm: middle-ironed clinching}, we compare the
performance between the DSIC mechanism and the welfare-optimal BIC
mechanism.

\begin{lemma}\label{cor: middle-ironed clinching}
There exists the instance of public-budget regular agents where 
the welfare-optimal DSIC mechanism is a $\gaplowerbound$-approximation to the welfare-optimal BIC mechanism.
\end{lemma}
\begin{proof}
Consider two agents with values drawn uniformly from $[0,\highestval]$ 
where $\highestval \geq 5.5$ 
and the budget $\budget = 1$.
By \Cref{thm: middle-ironed clinching},
the welfare-optimal DSIC mechanism in this case is the middle-ironed clinching auction
with $\valp = 0$ and $\valdp = 2$. 
The welfare-optimal BIC mechanism is the all-pay auction (applying the revelation principle). 
By computing the welfare for both mechanisms under this distribution, 
and setting $\highestval = 5.5$,
it optimizes the ratio as $\gaplowerbound$.
\end{proof}

%% file: DSIC-figure-alloc.tex
\begin{tikzpicture}[scale = 0.5]


\draw[->] (-0.2,0) -- (12.5, 0);
\draw[->] (0, -0.2) -- (0, 12.5);
\draw (3, 0.2) -- (3, -0.2);
\draw (6, 0.2) -- (6, -0.2);
\draw (12, 0.2) -- (12, -0.2);
\draw (0.2, 3) -- (-0.2, 3);
\draw (-0.2, 6) -- (0.2, 6);
\draw (-0.2, 12) -- (0.2, 12);

\draw (0, 12) -- (12, 12);
\draw (12, 12) -- (12, 0);
\draw (3, 3) -- (3, 12);
\draw (3, 3) -- (12, 3);
\draw (3, 6) -- (12, 6);
\draw (6, 3) -- (6, 12);
\draw (0,0) -- (3, 3);
\draw (6, 6) -- (12, 12);

\draw (0, -1.2) -- (0, -1.5);
\draw (3, -1.2) -- (3, -1.5);
\draw (6, -1.2) -- (6, -1.5);
\draw (12, -1.2) -- (12, -1.5);
\draw (0, -1.5) -- (1.2, -1.5);
\draw (1.8, -1.5) -- (3, -1.5);

\draw (0, -1.5) -- (1.1, -1.5);
\draw (1.9, -1.5) -- (3, -1.5);
\draw (3, -1.5) -- (4.1, -1.5);
\draw (4.9, -1.5) -- (6, -1.5);
\draw (6, -1.5) -- (8.6, -1.5);
\draw (9.4, -1.5) -- (12, -1.5);

\draw (0, -1.2) -- (0, -1.5);
\draw (3, -1.2) -- (3, -1.5);
\draw (0, -1.5) -- (1.2, -1.5);
\draw (1.8, -1.5) -- (3, -1.5);

\draw (0, -0.8) node {$0$};
\draw (3, -0.8) node {$\valp$};
\draw (6, -0.8) node {$\valdp$};
\draw (12, -0.8) node {$\highestval$};
\draw (13, -0.1) node {$\val_1$};

\draw (1.5, -1.5) node {$\lvalue$};
\draw (4.5, -1.5) node {$\mvalue$};
\draw (9, -1.5) node {$\hvalue$};

\draw (-1, 0) -- (-1.3, 0);
\draw (-1, 3) -- (-1.3, 3);
\draw (-1, 6) -- (-1.3, 6);
\draw (-1, 12) -- (-1.3, 12);

\draw (-1.3, 0) -- (-1.3, 1.1);
\draw (-1.3, 1.9) -- (-1.3, 3);
\draw (-1.3, 3) -- (-1.3, 4.1);
\draw (-1.3, 4.9) -- (-1.3, 6);
\draw (-1.3, 6) -- (-1.3, 8.6);
\draw (-1.3, 9.4) -- (-1.3, 12);

\draw (-1.3, 1.5) node {$\lvalue$};
\draw (-1.3, 4.5) node {$\mvalue$};
\draw (-1.3, 9) node {$\hvalue$};

\draw (-0.6, 0) node {$0$};
\draw (-0.6, 3) node {$\valp$};
\draw (-0.6, 6) node {$\valdp$};
\draw (-0.6, 12) node {$\highestval$};
\draw (0, 12.8) node {$\val_2$};

\draw (4.5, 1.5) node {\LARGE $1$};
\draw (1.5, 4.5) node {\LARGE $0$};
\draw (4.5, 4.5) node {\LARGE $\sfrac{1}{2}$};
\draw (9, 4.5) node {\LARGE $1$};
\draw (4.5, 9) node {\LARGE $0$};

\draw(9.9,7.3) node {\Large$\tfrac{1}{2} + \tfrac{\valp\valdp}{2\val_2^2}$};
\draw(8.4,10.7) node {\Large$\tfrac{1}{2} - \tfrac{\valp\valdp}{2\val_1^2}$};

\end{tikzpicture}

%% file: DSIC-figure-alloc2.tex
\begin{tikzpicture}[scale = 0.5]


\draw[->] (-0.2,0) -- (12.5, 0);
\draw[->] (0, -0.2) -- (0, 12.5);
\draw (4.5, 0.2) -- (4.5, -0.2);
\draw (12, 0.2) -- (12, -0.2);
\draw (0.2, 4.5) -- (-0.2, 4.5);
\draw (-0.2, 12) -- (0.2, 12);

\draw (0, 12) -- (12, 12);
\draw (12, 12) -- (12, 0);
\draw (4.5, 4.5) -- (4.5, 12);
\draw (4.5, 4.5) -- (12, 4.5);
\draw (0, 0) -- (12, 12);

\draw (0, -0.8) node {$0$};
\draw (4.5, -0.8) node {$\budget$};
\draw (12, -0.8) node {$\highestval$};
\draw (13, -0.1) node {$\val_1$};

\draw (-0.6, 0) node {$0$};
\draw (-0.6, 4.5) node {$\budget$};
\draw (-0.6, 12) node {$\highestval$};
\draw (0, 12.8) node {$\val_2$};

\draw (8, 2.25) node {\LARGE $1$};
\draw (2.25, 8) node {\LARGE $0$};

\draw(9.7,5.9) node {\LARGE $\tfrac{1}{2} + \tfrac{B^2}{2\val_2^2}$};
\draw(6.9,10.4) node {\LARGE $\tfrac{1}{2} - \tfrac{B^2}{2\val_1^2}$};

\end{tikzpicture}

%% file: DSIC-figure-partition.tex
\begin{tikzpicture}[scale = 0.5]

\fill[gray!60!white] (0, 0) -- (12, 0) -- (12, 3) -- (3, 3) -- (3, 12) -- (0, 12);
\fill[gray!20!white] (3, 3) -- (6, 3) -- (6, 6) -- (3, 6);
\fill[gray!40!white] (6, 6) -- (12, 6) -- (12, 12) -- (6, 12);

\draw[->] (-0.2,0) -- (12.5, 0);
\draw[->] (0, -0.2) -- (0, 12.5);
\draw (3, 0.2) -- (3, -0.2);
\draw (6, 0.2) -- (6, -0.2);
\draw (12, 0.2) -- (12, -0.2);
\draw (0.2, 3) -- (-0.2, 3);
\draw (-0.2, 6) -- (0.2, 6);
\draw (-0.2, 12) -- (0.2, 12);

\draw (0, 12) -- (12, 12);
\draw (12, 12) -- (12, 0);
\draw (3, 3) -- (3, 12);
\draw (3, 3) -- (12, 3);
\draw (3, 6) -- (12, 6);
\draw (6, 3) -- (6, 12);

\draw (0, -1.2) -- (0, -1.5);
\draw (3, -1.2) -- (3, -1.5);
\draw (6, -1.2) -- (6, -1.5);
\draw (12, -1.2) -- (12, -1.5);
\draw (0, -1.5) -- (1.2, -1.5);
\draw (1.8, -1.5) -- (3, -1.5);

\draw (0, -1.5) -- (1.1, -1.5);
\draw (1.9, -1.5) -- (3, -1.5);
\draw (3, -1.5) -- (4.1, -1.5);
\draw (4.9, -1.5) -- (6, -1.5);
\draw (6, -1.5) -- (8.6, -1.5);
\draw (9.4, -1.5) -- (12, -1.5);

\draw (0, -1.2) -- (0, -1.5);
\draw (3, -1.2) -- (3, -1.5);
\draw (0, -1.5) -- (1.2, -1.5);
\draw (1.8, -1.5) -- (3, -1.5);

\draw (0, -0.8) node {$0$};
\draw (3, -0.8) node {$\valp$};
\draw (6, -0.8) node {$\valdp$};
\draw (12, -0.8) node {$\highestval$};
\draw (13, -0.1) node {$\val_1$};

\draw (1.5, -1.5) node {$\lvalue$};
\draw (4.5, -1.5) node {$\mvalue$};
\draw (9, -1.5) node {$\hvalue$};

\draw (-1, 0) -- (-1.3, 0);
\draw (-1, 3) -- (-1.3, 3);
\draw (-1, 6) -- (-1.3, 6);
\draw (-1, 12) -- (-1.3, 12);

\draw (-1.3, 0) -- (-1.3, 1.1);
\draw (-1.3, 1.9) -- (-1.3, 3);
\draw (-1.3, 3) -- (-1.3, 4.1);
\draw (-1.3, 4.9) -- (-1.3, 6);
\draw (-1.3, 6) -- (-1.3, 8.6);
\draw (-1.3, 9.4) -- (-1.3, 12);

\draw (-1.3, 1.5) node {$\lvalue$};
\draw (-1.3, 4.5) node {$\mvalue$};
\draw (-1.3, 9) node {$\hvalue$};

\draw (-0.6, 0) node {$0$};
\draw (-0.6, 3) node {$\valp$};
\draw (-0.6, 6) node {$\valdp$};
\draw (-0.6, 12) node {$\highestval$};
\draw (0, 12.8) node {$\val_2$};

\draw (1.5, 1.5) node {\LARGE $\LS$};
\draw (4.5, 4.5) node {\LARGE $\MM$};
\draw (9, 9) node {\LARGE $\HH$};
\draw (9, 4.5) node {\LARGE $\HM$};
\draw (4.5, 9) node {\LARGE $\MH$};

\end{tikzpicture}

%% file: revelation_gap.tex
\section{Revelation Gap}
\label{s:revelation gap}
\label{s:gap}

In the literature, prior-independent mechanisms have been shown to
approximate the Bayesian optimal mechanism for many
objectives. Interestingly, except when the solution is trivial, none
of the approximation mechanisms in the literature are known to be
optimal. The formal question of optimal prior-independent mechanism
design is the following:
$$\beta = \min\limits_{\mech\in \MECHS} \max\limits_{F\in \DISTS}
\frac{\expect[\vals\sim F]{\OPT_F(\vals)}}{\expect[\vals\sim F]{\mech(\vals)}}.
$$
In this definition, $\expect[\val\sim F]{\mech(\val)}$ is the equilibrium performance of mechanism $\mech$ on distribution $F$ and $\OPT_F$ is the optimal mechanism for given objective on distribution $F$. 
Importantly in this definition, the family of mechanism $\MECHS$ may be restricted from all mechanisms and 
the family of distribution $\DISTS$ may be restricted from all distributions.

As discussed in the introduction, the revelation principle is not without loss for prior-independent mechanism design. 
Based on this idea, we introduce the concept of revelation gap.
\begin{definition}
The \emph{revelation gap} is the ratio of the prior-independent approximation of incentive compatible mechanisms to the prior-independent approximation of (generally non-revelation) mechanisms.
\end{definition}

In this section, we consider welfare maximization with public-budget
regular agents. With the results in the previous sections, we give
both the upper bound and the lower bound of revelation gap for this
objective.
\begin{theorem}
For public-budget regular i.i.d.\ agents, the revelation gap for welfare maximization is at most $e$.
Specifically, this upper bound considers prior-independent DSIC, ex-post IR mechanisms.
\end{theorem}
\begin{proof}
This upper bound is given by considering the clinching auction which is a 
prior-independent DSIC and ex-post IR mechanism.
\Cref{thm:clinching auction e approx} says that the clinching auction is 
an $e$-approximation to the welfare-optimal mechanism for public-budget regular agents.
Thus, the revelation gap is at most $e$.
\end{proof}

For the lower bound, we use the result in \Cref{s:DSIC} where we
solve the welfare-optimal DSIC mechanism for two agent with uniformly
distributed values.  Note that for two-agent environments, the DSIC
ex-post IR constraints are equivalent to prior-independent BIC and IIR
constraints.\footnote{Prior-independent BIC and IIR mechanisms are the
  mechanisms which are BIC and IIR for all i.i.d.\ distributions.
  This property is stronger than BIC (for a single distribution) but
  generally weaker than DSIC.}  With more than two agents, this
equivalence does not generally hold.

\begin{lemma}
For two i.i.d.\ agents, a mechanism is Bayesian incentive compatible
and interim individual rational for all i.i.d.\ distributions if and
only if it is dominant strategy incentive compatible and ex-post
individually rational.
\end{lemma}
\begin{proof}
The direction that
DSIC implies BIC for all i.i.d.\ distribution is trivial by the definition.
To show the other direction, for arbitrary value $\val$, consider the distribution 
which puts the whole mass on $\val$. 
These distributions break the interim constraints in BIC into
the ex-post constraints in DSIC for every value profiles.
Hence, BIC for all i.i.d.\ distribution implies DSIC for two agents setting.
\end{proof}

\begin{theorem}
For public-budget regular i.i.d.\ agents, the revelation gap for welfare maximization is at least $\gaplowerbound$.
\end{theorem}
\begin{proof}
This lower bound is given by considering the all-pay auction
and the middle-ironed clinching auction.

As the characterization in \Cref{s:prelim}, the all-pay auction is 
a prior-independent mechanism. \Cref{thm: all-pay regular} says that
the all-pay auction is welfare-optimal for public budget regular agents.
Hence, the prior-independent approximation of
the all-pay auction is 1.

Next, we show that the prior-independent approximation of Bayesian incentive compatible mechanisms is at least $\gaplowerbound$.
\Cref{thm: middle-ironed clinching} says that
the middle-ironed clinching auction is Bayesian optimal DSIC mechanism for 
two agents with values drawn uniformly from $[0,\highestval]$.
Since for two agents case, 
the DSIC property is equivalent to  
the BIC for all i.i.d.\ distribution property,
\Cref{cor: middle-ironed clinching} suggests that
the prior-independent approximation of incentive compatible mechanisms is at least 
$\gaplowerbound$.

Thus, the revelation gap for welfare maximization is at least $\gaplowerbound$.
\qedhere

\end{proof}

%% file: irregular.tex
\section{Welfare Approximation for Irregular Distribution}
\label{s:irregular}

In this section, we analyze the welfare approximation of the all-pay auction and the clinching auction 
for public budget agents without regularity assumption.

The main technique we use is the 
following lemma which relaxes the budget constraint to another constraint which upper bounds the wining probability of the highest value,
i.e.\ $\alloc(\highestval)$
where $\highestval$ is the highest value in the support of the distribution.

\begin{lemma}\label{lem:relax budget constraint}
Given any interim constraint $\calloc$ and budget $\budget$,
let $\valp$ be the value where the budget binds 
in $\calloc$ after ironing from $\valp$ to $\highestval$,
i.e.\ $\valp \cdot \starred \aalloc(\valp) - \int_{0}^{\valp} \calloc(t)dt = \budget$ where $\starred \aalloc(\valp) = \frac{1}{1-F(\valp)}\int_{\valp}^\highestval \starred \alloc(t)dF(t)$,
the averaging wining probability for value beyond $\valp$ in allocation $\calloc$. 
Any interim feasible and budget balanced allocation $\alloc$ 
satisfies 
$$\alloc(\highestval) \leq 2\starred \aalloc(\valp).$$

\end{lemma}

\begin{proof}
Recall that $\val\primed $ is the value where the budget binds in $\calloc$ after ironing from $\valp$ to $\highestval$.
Thus,
\begin{align*}
\budget &= \valp \cdot \starred \aalloc(\valp) - \int_{0}^{\valp} \calloc(t)dt
\leq \valp \cdot \starred \aalloc(\valp).\\
\intertext{On the other hand, suppose $\alloc$ is budget balance,}
\budget &\geq
\highestval \cdot \alloc(\highestval) - 
\int_{0}^{\highestval}\alloc(t)dt
\geq 
\valp\cdot(\alloc(\highestval) - \alloc(\valp))\\
\intertext{Suppose $\alloc$ is interim feasible, }
\alloc(\valp) &\leq 
\frac{1}{1 - F(\valp)}\int_{\valp}^\highestval \alloc(t)dt 
\leq 
\frac{1}{1 - F(\valp)}\int_{\valp}^\highestval \calloc(t)dt 
=
\starred \aalloc(\valp).\\
\intertext{Combine the inequalities above,}
\alloc(\highestval)
&\leq
\alloc(\valp) +\starred \aalloc(\valp)
\leq
2\starred \aalloc(\val \primed).
\qedhere
\end{align*}
\end{proof}

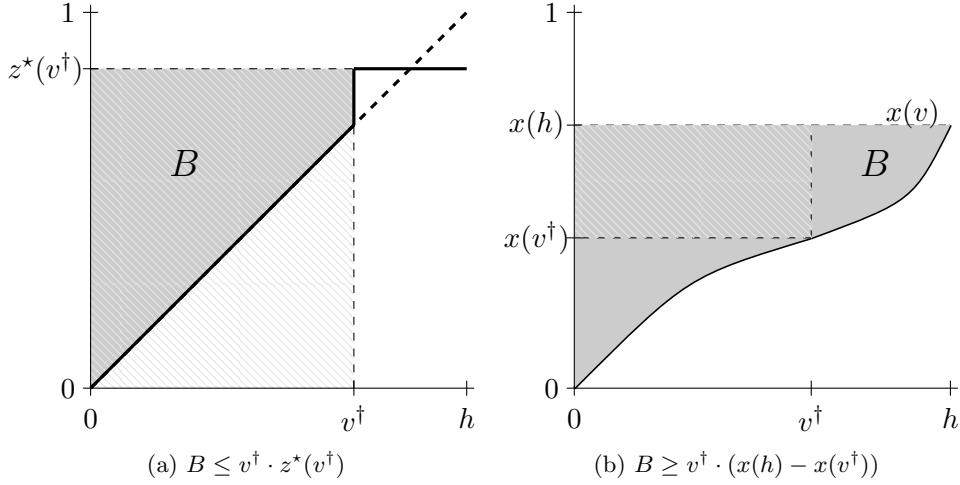
\begin{figure}
\centering
\subfloat[$\budget\leq \valp\cdot 
\starred \aalloc(\valp)$]
{\input{irregular-figure-1}}
\subfloat[$\budget\geq \valp\cdot 
(\alloc(\highestval) - \alloc(\valp))$]
{\input{irregular-figure-2}}
\caption{Proofs by picture of the upper bound
and lower bound on budget $\budget$.}
\end{figure}

\subsection*{The All-pay Auction}
First, we discuss the performance of all-pay auction for the irregular distribution. 
\citet{PV-14} show that the welfare-optimal interim allocation is both ironing top interval and perhaps ironing some other intervals in the middle. It turns out that even though the all-pay auction only irons the top interval, its welfare only suffers a modest loss.

\begin{theorem}\label{thm:all-pay 2-approx}
For public-budget i.i.d.\ agents, the all-pay auction is a 2-approximation to the welfare-optimal mechanism.
\end{theorem}

\begin{proof}

Applying Lemma \ref{lem:relax budget constraint}, we relax the budget constraint to the constraint that $\alloc(\highestval) \leq 2\starred\aalloc(\val\primed)$.

Denote $\alloc^0$ as the welfare-optimal interim feasible and budget balanced allocation and $\alloc$ as the welfare-optimal interim feasible allocation under the relaxed constraint,
then $\welfare[\alloc^0] \leq \welfare[\alloc]$.

Since $\alloc$ maximizes the welfare under interim feasibility 
constraint $\calloc$, it allocates as by $\calloc$ 
except that values in $[\valdp,\highestval]$ are ironed. 
The threshold $\valdp$ is selected such that 
$\alloc(\highestval) = 2\starred\aalloc(\valp)$.
By definitions of $\valp$ and $\valdp$, we know $\valp \leq \valdp$.
Consider $\alloc$ for values below and beyond $\valp$ separately.
For value below $\valp$, the expected welfare
$\int_0^{\valp}\val\cdot\alloc(\val)dF(\val)
=
\int_0^{\valp}\val\cdot\calloc(\val)dF(\val).
$
For value beyond $\valp$,
the expected welfare
$\int_{\valp}^\highestval\val\cdot\alloc(\val)dF(\val)
\leq
\int_{\valp}^\highestval\val\cdot 2\starred\aalloc(\valp)dF(\val).
$

Notice that $\valp$ coincides with the threshold in the all-pay auction,
and the all-pay auction allocates as by $\calloc$ 
except value beyond $\valp$ win with probability 
$\starred\aalloc(\valp)$. 
Thus,
$\welfare[\alloc] \leq 2\cdot \welfare[\text{All-pay}]
$.
\end{proof}

In fact, the 2-approximation bound is tight.

\begin{lemma}
There exists the instance where the welfare of the all-pay auction is half of welfare-optimal mechanism.
\end{lemma}
\begin{proof}
Consider the following single-item instance with budget $\budget = 1$. 
There are  $N + 1$ agents with valuation distribution
\begin{align*}
\begin{split}
\val = \left\{
\begin{array}{ll}
N - \epsilon & \qquad \text{w.p. } 
\frac{1}{N + 1},\\
N & \qquad \text{w.p. } 
\frac{N - 1}{N + 1},\\
N^3 & \qquad \text{w.p. } 
\frac{1}{N + 1}.\\
\end{array}
\right.
\end{split}
\end{align*}

In the all-pay auction, the interim allocation rule irons values $N$ and $N^3$, 
\begin{align*}
\begin{split}
x(v) = \left\{
\begin{array}{ll}
\delta & \quad \text{if } 
v = N - \epsilon,\\
\frac{1 - \delta}{N} & \quad \text{if } 
v = N \text{ or } N^3,\\
\end{array}
\right.
\quad \text{
where $\delta = (\tfrac{1}{(N + 1)})^{N + 1} \rightarrow 0$.}
\end{split}
\end{align*}
where the expected welfare is roughly $N + 1$.

Notice that in the all-pay auction, the mechanism uses almost all 
budget to distinguish values $N -\epsilon$ and $N$ whose contribution
to the expected welfare is almost the same.
Therefore, consider the auction which irons values $N - \epsilon$ and $N$, and moves some winning probability from value $N^3$ to values $N-\epsilon, N$ as follows,
\begin{align*}
\begin{split}
x(v) = \left\{
\begin{array}{ll}
\frac{N - 1}{N(N + 1)} & \qquad \text{if } 
v = N - \epsilon \text{ or } N,\\
\frac{2}{N + 1} & \qquad \text{if } 
v = N^3.\\
\end{array}
\right.
\end{split}
\end{align*}
where the expected welfare is roughly $2N + 1$.

Let $N \rightarrow \infty$ and $\epsilon \rightarrow 0$, the expected welfare from the all-pay auction is exactly half of the expected welfare from the optimal auction.
\end{proof}

\subsection*{The Clinching Auction}
To analyze the welfare approximation of the clinching auction for
irregular distributions, we follow almost the same argument as for
regular distributions.  The only difference in the argument is that
the ex ante welfare-optimal mechanism may not be a simple posted price
for a probabilistic allocation.  However, by Lemma~\ref{lem:relax
  budget constraint}, such a posted pricing is still a
2-approximation.
\begin{lemma}\label{lem:PO approx}
For public-budget i.i.d.\ agents, the posted pricing described in
Lemma~\ref{lem:PO} is a 2-approximation to the ex ante welfare-optimal
mechanism.
\end{lemma}
\begin{proof}
Consider the interim constraint $\calloc(\val) = 1$ if $\dist(\val)
\geq 1/n$ and $\calloc(\val) = 0$ if $\dist(\val) < 1/n$. Applying the similar argument in
Theorem~\ref{thm:all-pay 2-approx} with Lemma~\ref{lem:relax budget constraint}, the lemma holds.
\end{proof}

The following corollary is combines Lemma~\ref{lem:PO approx} with
\autoref{thm:clinching auction e approx}.

\begin{corollary}\label{cor:clinching auction 2e approx}
For public-budget i.i.d.\ agents, the clinching auction is a
$2e$-approximation to the welfare-optimal mechanism.
\end{corollary}


%% file: irregular-figure-1.tex
\begin{tikzpicture}[scale = 0.5]
\fill[gray!40!white] (0, 0) -- (7, 7) -- (7, 8.5) -- (0, 8.5);

\fill[pattern=north west lines, pattern color=gray!20!white] (0, 0) -- (0, 8.5) 
-- (7, 8.5) -- (7, 0);

\draw (2.5, 6) node {\Large$\budget$};

\draw (-0.2,0) -- (10.2, 0);
\draw (0, -0.2) -- (0, 10.2);
\draw (7, 0.2) -- (7, -0.2);
\draw (10, 0.2) -- (10, -0.2);
\draw (-0.2, 8.5) -- (0.2, 8.5);
\draw (-0.2, 10) -- (0.2, 10);

\draw [dashed] (0, 8.5) -- (7, 8.5);
\draw [dashed] (7, 0) -- (7, 7);

\begin{scope}[very thick]

\draw (0, 0) -- (7, 7);
\draw (7, 7) -- (7, 8.5);
\draw (7, 8.5) -- (10, 8.5);
\draw [dashed] (0,0) -- (10, 10);
\end{scope}

\draw (0, -0.8) node {$0$};
\draw (7, -0.8) node {$\valp$};
\draw (10, -0.8) node {$\highestval$};

\draw (-0.6, 0) node {$0$};
\draw (-1.2, 8.5) node {$\starred \aalloc(\valp)$};
\draw (-0.6, 10) node {$1$};

\end{tikzpicture}

%% file: irregular-figure-2.tex
\begin{tikzpicture}[scale = 0.5]

\begin{scope}[very thick]
\draw [name path global = A](0, 0) .. controls 
(3, 3) .. (6.3, 4);
\draw [name path global = B](6.3, 4) .. controls
(9, 5) .. (10, 7);

\end{scope}

\draw [name path = C, dashed]
(0, 7) -- (6.3, 7);
\draw [name path = D, dashed]
(6.3, 7) -- (10, 7);

\tikzfillbetween[of=A and C]{gray!40!white};
\tikzfillbetween[of=B and D]{gray!40!white};

\draw [dashed] (6.3, 4) -- (0, 4);
\draw [dashed] (6.3, 4) -- (6.3, 7);
\fill[pattern=north west lines, pattern color=gray!20!white] (0, 4) -- (0, 7) 
-- (6.3, 7) -- (6.3, 4);



\draw (8, 6) node {\Large$\budget$};

\draw (-0.2,0) -- (10.2, 0);
\draw (0, -0.2) -- (0, 10.2);
\draw (6.3, 0.2) -- (6.3, -0.2);
\draw (10, 0.2) -- (10, -0.2);
\draw (-0.2, 4) -- (0.2, 4);
\draw (-0.2, 10) -- (0.2, 10);
\draw (-0.2, 7) -- (0.2, 7);


\draw (0, -0.8) node {$0$};
\draw (6.3, -0.8) node {$\valp$};
\draw (10, -0.8) node {$\highestval$};

\draw (-0.6, 0) node {$0$};
\draw (-1.0, 4) node {$\alloc(\valp)$};
\draw (-1.0, 7) node {$\alloc(\highestval)$};
\draw (-0.6, 10) node {$1$};

\draw (9, 7.3) node {$\alloc(\val)$};
\end{tikzpicture}

%% file: winner.tex
\section{Welfare Approximation of Losers-pay-nothing Mechanisms}
\label{s:winner}

Both the all-pay and clinching auctions discussed previously allocate
probabilistically and agents make payments even when they lose.  These
mechanisms contrast with the first-price and second-price auctions
which have the property that losers pay nothing.  In this section we
will show that losers-pay-nothing mechanisms do not give good
approximations to the welfare-optimal mechanism.

Our approach is to prove three results.  First, we show that when
optimizing over losers-pay-nothing mechanisms, it is without loss in
the objective to consider only winner-pays-bid mechanisms.  Second, we
show that for public-budget regular agents and a single-item the the
welfare-optimal winner-pays-bid mechanism is highest bid wins, i.e.,
it is the first-price auction. Third, we show that there is a
public-budget regular distribution where the welfare of the
first-price auction is a linear factor from the welfare of the optimal
mechanism (which is the all-pay auction).  Combining these results we
see that losers-pay-nothing mechanisms are not good approximation
mechanisms.

The main difference between winner-pays-bid and all-pay mechanisms is
that the budget binds over a larger range of values for
winner-pays-bid mechanisms.  Consequently the allocation rule of
winner-pays-bid mechanisms is further from the efficient
highest-value-wins allocation rule.

\subsection*{Winner-pays-bid versus Loser-pays-nothing Mechanisms}

Mechanisms can map bids to probabilistic outcomes.  Denote the random
bits accessed by a mechanism by $\rand$.  Our previously defined
allocation and payment rules take expectation over this randomization,
i.e., $\ballocs(\bids) = \expect[\rand]{\ballocs^\rand(\bids)}$ and
$\bprices(\bids) = \expect[\rand]{\bprices^\rand(\bids)}$ where both
$\ballocs^{\rand}$ and $\bprices^{\rand}$ are deterministic functions.
Specifically $\ballocs^{\rand}(\bids) \in \{0,1\}^n$.

\begin{definition}
A \emph{loser-pays-nothing} mechanism $(\ballocs,\bprices)$ satisfies
$\bpricei^{\rand}(\bids) = 0$ if $\balloci^{\rand}(\bids) = 0$ for all
agents $i$ and random bits $\rand$.
\end{definition}

\begin{definition}
A \emph{winner-pays-bid} mechanism $(\ballocs,\bprices)$ satisfies
$\bpricei^{\rand}(\bids) = \bidi\,\balloci^{\rand}(\bids)$ for all
agents $i$ and random bits $\rand$.
\end{definition}

\begin{lemma}
For public budget agents, a winner-pays-bid mechanism is optimal among
all loser-pays-nothing mechanisms.
\end{lemma}

\begin{proof}
Let $(\allocs,\prices)$ be the BNE allocation and payment rules of any
loser-pays-nothing mechanism.  First, disregarding the budget
constraints of the agents, there is a winner-pays-bid mechanisms with
the same BNE allocation and payment rules.  This fact is a
straightforward consequence of the characterization of Bayes-Nash
equilibrium, e.g., see \citet{CH-13}.  Second, as these two mechanisms
have the same interim allocation rule, namely $\allocs$; the
payment-identity requires that they have the same interim payment
rule, namely $\prices$.  Losers pay nothing in both mechanisms.  Thus,
the expected payment of an agent $i$ with a value $\vali$ conditioned
on winning is the same in the two mechanisms.  In the winner pays bid
mechanism the agent's payment conditioned on winning is
deterministically equal to its conditional expectation.  In the
original loser-pays-nothing mechanism the maximum payment in the
support of the conditional payment distribution is no lower than its
expectation.  Consequently, the disregarded budget constraints are
satisfied by the constructed winner-pays-bid mechanism and it obtains
the same objective value.
\end{proof}

\subsection*{Optimal Winner-pays-bid Mechanisms}

Recall that for i.i.d.\ public-budget regular agents, the optimal
mechanism allocates efficiently except that an interval of
highest-valued agents are ironed and payments are deterministic
functions of values.  Specifically, in a single-item environment, the
all-pay auction is optimal.  We will show below that, restricting the
mechanism to be winner-pays-bid, the optimal mechanism has the same
form.  It allocates efficiently except that an interval of highest
values agents are ironed.  Specifically, in a single-item environment,
the first-price auction is optimal.  Notice that the first-price
auction is the winner-pays-bid highest-bid-wins mechanism.

\begin{theorem}\label{thm:first-price}
For i.i.d.\ public-budget regular agents, the welfare-optimal
winner-pays-bid mechanism is the highest-bid-wins mechanism, i.e.\ the
first-price auction.
\end{theorem}

The formal proof of Theorem~\ref{thm:first-price} is given in
Appendix~\ref{app:winner}.  It can be proved using standard methods in
Bayesian mechanism design with budgets, e.g., \citet{LR-96} and
\citet{mas-00}.  We give simpler approach that is based on the
\citet{AFHH-13} reduction to ex ante pricing.  Note that the approach
of \citet{AFHH-13} cannot be directly applied as public-budget agents
do not have linear utility and the optimal revenues for non-linear
agents do not satisfy an important linearity property.  We write the
problem as a optimization program, use the Lagrangian relaxation to
move the budget constraint into the objective, and then observe that
the Lagrangian relaxed objective does satisfy the linearity property
of \citet{AFHH-13}.

\subsection*{Lowerbound for Winner-pays-bid Mechanisms}

With Theorem~\ref{thm:first-price}, we compare the performance of the
welfare-optimal winner-pays-bid mechanism and the welfare-optimal
mechanism.  For i.i.d.\ public-budget regular agents and a
single-item, this comparison is between the first-price auction and
the all-pay auction.

\begin{lemma} For $n$ i.i.d.\@ public-budget agents, the first-price auction is at best an $\left(\frac{1}{4}n-o(n)\right)$-approximation to the all-pay auction.
\end{lemma}

\begin{proof}
Consider $n$ agents with budget $\budget = (1-\frac{1}{e})\frac{1}{n}$
and value distributed from the following distribution $F$ with density
function
\begin{align*}
f(\val) = \begin{cases}
n - 1 \quad &\text{if } \val \in [0, \frac{1}{n}] \\
\frac{1}{n-1} &\text{if } \val \in (\frac{1}{n}, 1].
\end{cases}
\end{align*}

By Theorem~\ref{thm:first-price}, the welfare-optimal winner-pays-bid
mechanism is the first-price auction where each agent bids $\budget$
if her value is beyond some $\valp$ to satisfy
$\frac{\price_{\text{first-price}}(\valp)}{\alloc_{\text{first-price}}(\valp)}=
\budget$.  In this setting, as $n$ goes to infinity, $\valp$ goes to
$\budget$, and the expected welfare goes to $\left((1 -
\frac{1}{2e})(1 - \frac{1}{e}) + \frac{1}{2}\right)\frac{1}{n} +
o(\frac{1}{n})$.

The welfare-optimal mechanism is the all-pay auction where each agent
bids $\budget$ if her value is beyond some $\valdp$ such that
$\price_{\text{all-pay}}(\valdp) = \budget$.  In this setting, as $n$
goes infinity, $\valdp$ goes to $\frac{1}{n}$, and the expected
welfare goes to $(1 - \frac{1}{e})\frac{1}{2} + o(1) $.

Thus, as $n$ goes to infinity, the expected welfare of the first-price
auction is an $\frac{1}{4}n$-approximation to the all-pay auction.
\end{proof}

It is easy to see that this linear approximation is tight up to
constant factors.  Specifically, the {\em $n$-agent lottery}, which
allocates the item to a random agent without payments, is trivially an
$n$-approximation.


%% file: revenue.tex
\section{Revenue Approximation of the All-pay Auction}
\label{s:revenue}

In this section, we analyze the approximation ratio of the all-pay
auction for public-budget regular agents with the revenue-optimal
mechanism.  In the preliminaries, Theorem \ref{thm: opt revenue} shows
that the revenue-optimal mechanism irons the top and sets a reserve
with the combined effect that the budget binds for the highest value.
From the equivalence of the all-pay auction (\autoref{thm: all-pay
  regular}) and the welfare-optimal auction (\autoref{thm: opt
  welfare}), the all-pay auction irons top values to decrease the
payment of highest value to meet the budget.  Even though the value
intervals ironed at top are different between all-pay auction and
optimal auction (specifically the all-pay auction irons less than the
optimal auction), the all-pay auction is still a good approximation to
the revenue-optimal mechanism.  Our analysis follows
\citeauthor{kir-06}'s \citeyear{kir-06} proof of the main theorem from
\citet{BK-96}.  These results show that, without budgets, the revenue
of the second-price auction approximates the optimal revenue for
i.i.d.\ and regular agents.  The revenue-optimal mechanism allocates
the item to the highest agent whose value exceeds a specific reserve
price while the second-price auction allocates the item to the highest
agent with no reserve.  As is crucial for \citeauthor{kir-06}'s proof,
the second-price auction can be thought as the revenue-optimal
mechanism under the constraint that the item must be allocated.

\begin{theorem}[\citealt{BK-96}]
For $n \geq 2$ regular i.i.d.\ agents with linear utilities, 
the expected revenue of the second-price auction is an $\frac{n}{n-1}$-approximation 
to the revenue-optimal mechanism.
\end{theorem}

For agents with public budget, the all-pay auction plays the similar role as the second-price auction.

\begin{theorem}
For $n \geq 2$ public-budget regular i.i.d.\ agents, the all-pay auction is an $\tfrac{n}{n-1}$-approximation to the revenue-optimal mechanism.
\end{theorem}

\begin{proof}

First, we introduce a common approach of the revenue analysis in Bayesian mechanism design.
\citet{mye-81} defined the virtual valuation for sum of payments as
$\phi(\val) = \val - \frac{1-F(\val)}{f(\val)}$ and proved that 
the expected payment $\expect{\price(\val)}$ of an agent in any BNE is equal
to her expected virtual surplus 
$\expect{\phi(\val)\alloc(\val)}$.
We use this concept in the following argument.

Denote value $\valp$ as the threshold where the all-pay auction starts to iron.

Consider the optimal mechanism which is (i) budget balanced, 
(ii) irons the interval $[\valp, \highestval]$, and 
(iii) always sells the item for $n$ public-budget regular i.i.d.\ agents.

We claim that the optimal auction under these three requirements is
the all-pay auction by analyzing the virtual surplus.\footnote{Under
  Requirement (i) and Requirement (iii), the all-pay auction is
  already optimal.  We introduce Requirement (ii) to simplify the
  argument.  } For public-budget regular agents, their virtual value
is monotone non-decreasing.  Consider the virtual surplus from values
$\val$ below and above $\valp$ separately.  For values $\val \geq
\valp$, requirement (ii) upper bounds the allocation $\alloc(\val)$ to
be at most as the all-pay auction (otherwise, the budget constraint
will be violated).  For values $\val \leq \valp$, the all-pay auction
always allocates the item to the agent with highest virtual value.
Thus, the all-pay auction maximizes the virtual surplus under these
three requirements.

On the other hand, consider the following auction LB: 
\begin{enumerate}
\item run the $n$-agent revenue-optimal budget-balanced auction on the
  first $n-1$ agents and a fake agent;
\item if the auction does not sell the item to anyone within the first $n-1$ agents 
(i.e.\ the auction does not sell the item or sells it to the fake agent), 
give the item to the $n$-th agent for free.
\end{enumerate}
LB is budget balanced and always sells the item for $n$ public-budget regular i.i.d.\ agents.
Notice that the revenue-optimal auction for public budget $n$ agents irons more than
the all-pay at the top. 
Hence, LB satisfies the three requirements.
The expected revenue from it is $\tfrac{n-1}{n}$ fraction of the expected revenue from the revenue-optimal auction for $n$  public-budget regular i.i.d.\ agents.

We conclude that the expected revenue of the all-pay auction is at least the expected revenue of  LB  
which is $\tfrac{n-1}{n}$ fraction of the revenue-optimal mechanism. 
Thus, the all-pay auction is an $\tfrac{n}{n-1}$ approximation to the revenue-optimal mechanism.
\end{proof}

%% file: appendix-price-jump.tex
\section{Clinching Auction with Price Jumps}

\label{app:price-jump}

In this section, we introduce the clinching auction with price jumps.
In the standard clinching auction, with a continuous increasing price-clock,
excess demand decreases continuously to the point where supply
equals demand and the market clears.  With a price jump, which leads
to a strict drop of demands, the standard clinching auction may leave
some supply unallocated.  Therefore, to clear the market, the
clinching auction with price jumps will need to reallocate some amount
of units at the pre-jump price after a price jump.  We
first focus on the clinching auction with price jumps for agents with
identical budgets.  The results can be extended to agents with distinct
budgets, which we will discuss at the end of this section.

To formally describe this reallocation, suppose the price-clock jumps
from $\valp$ to $\valdp$.  Consider the state $\mathcal C = (s, \Sp,
\budget)$ at price $\valp$ after the clinching step (i.e.\ Step~3 in
\Cref{def:clinching}) where $s$ is the current supply remaining, $\Sp$
is the agents with values at least $\valp$ (let $\kp = |\Sp|$),
and $\budget$ is the current budget of the active agents.
\footnote{In the model considered in this paper where initially the
  agents have identical budgets, the remaining budgets of all active agents
  remain identical throughout the execution of the clinching auction.}
When the price jumps to $\valdp$, active agents with values below
$\valdp$ (``low-valued'' agents) will quit, while active agents
with values at least $\valdp$ (``high-valued'' agents) will stay in
the auction.  Denote by $\Sdp$ the set of high-valued agents and by $k
= |\Sdp|$ the number of high-valued agents.  With pre-jump state
$\mathcal C = (s, \Sp, \budget)$ and $k$ agents remaining
after the jump, define $h^{\mathcal C}_k$ and $l^{\mathcal C}_k$ as
the additional supply allocated at the low price to high- and
low-valued agents, respectively.  

In the following discussion, we fix an arbitrary state $\mathcal C =
(s, \Sp, \budget)$ with $\kp = |\Sp|$ active agents, drop the
superscript of $h^{\mathcal C}_k$ and $l^{\mathcal C}_k$, and consider
$h_k$ and $l_k$
constrained to the following polytope:
\begin{align}\label{price jump}
\begin{split}
\begin{aligned}
\text{IC: } &\forall
k \in \{1, \dots, \kp\}
&h_k &= l_{k - 1},\\
\text{BB: } &\forall
k\in\{0, \dots, \kp\}
&h_k,l_k &\leq \budget/\valp,\\
\text{NN: } &\forall
k \in \{0,\dots, \kp\} 	
&h_k,l_k &\geq 0, \qquad \\
\text{MC: } &\forall 
k \in \{0,\dots, \kp\}& 
k \, h_k + (\kp - k) \, l_k
+ \tfrac{k}{\valdp}(\budget - \valp\, h_k)
&\geq s,
\\
\text{LS: } &\forall
k \in \{0, \dots, \kp\}
&k\, h_k + (\kp-k)\, l_k& \leq s. \\
\end{aligned}
\end{split}
\end{align}
The constraints above are, respectively, incentive compatibility (IC),
budget balance (BB), non-negative consumption (NN), market clearing
(MC), and limited supply (LS).  The IC constraint requires that the
amount of supply which an agent gets at price $\valp$ does not depend
on whether the agent stays or quits during the price jump.  The
left-hand side of the incentive compatibility (IC) constraint is the
additional allocation quantity at price $\valp$ if an agent stays
during the price jump, while the right hand side is the additional
allocation quantity at price $\valp$ if she quits.  Since the two
quantities are equal, active agents with values in $[\valp,\valdp)$
  prefer to quit after the clinching step at price $\valp$ while
  agents with value at least $\valdp$ prefer to stay in the auction.
  The market clearing (MC) constraint states that the reallocated supply at
  the low price $\valp$ (the first two terms) plus the quantity
  demanded by the high-valued agents at the high price $\valdp$ (the
  third term) must be at least the supply.  The limited supply (LS)
  constraint states that the amount allocated at the low price for any
  number $k$ of high-valued agents must not exceed the supply.

Consider the problem of selecting a point in polytope \eqref{price
  jump} to optimize the expected welfare under the value distribution
$\dist$.  First notice that, since the state $\mathcal C = (s, \Sp,
\budget)$ is induced by the clinching auction, the total demand at
price $\valp$ under budget $\budget$ exceeds the remaining supply $s$,
i.e., $\kp \budget / \valp \geq s$; setting $h_k = l_k = s/ \kp$
for all $k\in[\kp]$ is feasible; and, thus, polytope \eqref{price
  jump} is not empty.  The expected welfare of the clinching auction
is complicated to express; we instead consider the objective of
minimizing, within the constraints of polytope \eqref{price jump}, the
expected supply reallocated to low-valued agents, i.e.,
$\sum_{k=0}^{\kp} l_k \pi^{\kp-k} (1-\pi)^k$ 
where $\pi =
\frac{\dist(\valdp)-\dist(\valp)}{1 - \dist(\valp)} 
= \prob[\val \sim \dist]{\val < \valdp
  \mid \val \geq \valp}$ is the probability an agent has a low
value. Based on this reallocation, we formally define a clinching
auction with price jumps and show that it clears the market, is
ex-post IR, and is DSIC.

\begin{definition}
The \emph{clinching auction with price jumps} maintains an allocation
and price-clock starting from zero.  Before and after each price jump
point, the price-clock ascends continuously and the allocation and the
budget are adjusted as in the standard clinching auction.  When the
price-clock jumps from $\valp$ to $\valdp$ the following steps are taken:
\begin{enumerate}
\item run the standard clinching steps on price-clock $\valp\!$ and
  the current budgets and let the subsequent state be $\mathcal C = (s,
  \Sp\!, \budget)$ with $\kp = |\Sp|$;
\item increase the price-clock to $\valdp$ and let $\kdp = |\Sdp|$ be the
  number of agents remaining in the auction;
\item solve for $\{h_k, l_k\}_{k \in [\kp]}$ to minimize 
the expected quantity reallocated to the low-valued agents 
in the polytope \eqref{price jump};
\item allocate $h_{\kdp}$ units at price $\valp$ to each of the $\kdp$ agents
  that stay after the price jump, allocate $l_{\kdp}$ units at price
  $\valp$ to each of the $\kp-\kdp$ agents that quit during the price jump,
  and adjust all the agents' budgets for the amount and price
  allocated;
\item run the standard clinching step with price-clock $\valdp\!$ and updated budgets.
\end{enumerate}
\end{definition}

\begin{proposition}
The clinching auction with price jumps always clears the market.
\end{proposition}
\begin{proof}
If the price-clock increases continuously, the demand decreases
continuously.  When the total demands meet the supply remaining,
\citet{DLN-08} show that the standard clinching auction halts and the
market clears.  For the clinching auction with price jumps, when the
price-clock goes through a price jump, the market clearing constraints
that define polytope \eqref{price jump} guarantee that the total
demands are at least the supply remaining.  Thus, the clinching
auction with price jumps clears the market.
\end{proof}

\begin{proposition}
The clinching auction with price jumps satisfies ex-post IR, DSIC,
and budget balance.
\end{proposition}
\begin{proof}
\citet{DLN-08} show that the standard clinching auction is ex-post IR,
DSIC, and budget balanced.  For the clinching auction with price
jumps, when the price-clock goes through a price jump from $\valp$ to
$\valdp$, the IC constraints that define polytope \eqref{price jump}
guarantee that the agents with values at most $\valdp$ weakly prefer to
quit at price $\valp$ and the agents with values above $\valdp$ prefer
to stay at price $\valp$.  Meanwhile, the budget constraints and
non-negative consumption constraints that define polytope 
\eqref{price jump} guarantee that the agents are budget balanced and have
non-negative utility after the price jump.
\end{proof}

For two i.i.d.\ agents with identical budgets, the clinching auction
with price jumps induces the same outcome as the middle-ironed
clinching auction (\Cref{d:middle-ironed-clinching-auction}).  For a general number of agents, it is polynomial
time solvable.  We conjecture that, for i.i.d.\ distributions and
identical budgets, minimizing the expected quantity reallocated to
low-valued agents, i.e., the objective described previously, is
equivalent to maximizing expected welfare.  We leave to future studies
the question of whether there is a more succinct characterization of
the expected welfare maximizing solution and the generalization to agents
with non-identical valuation distributions.

If agents have distinct budgets, the linear program can be generalized
by replacing the variables, which corresponded to the reallocation to
high- and low-valued agents with a given number $k=|\Sdp|$ of
high-valued agents, with variables that correspond to the reallocation
to each agent $i$ with a given set $\Sdp$ of high-valued agents.  With
this modification to the variables and constraints of polytope
\eqref{price jump}, the previous argument guarantees the new polytope
is non-empty.  Notice that there are $O(n\cdot 2^{n})$ variables
defining the new polytope.  It is possible, however, to optimize
expected allocation to the low-valued agents subject to this polytope
in polynomial time when there are a constant number of distinct
budgets; symmetries across agents with identical budgets allow the
number of variables in the program to be reduced to a polynomial
number.  We leave to future studies the problem of identifying a
polynomial time algorithm for optimally reallocating the supply during
a price jump when there are generally distinct budgets.

%% file: appendix-winner.tex
\section{Bayesian Optimal Mechanisms for Budgeted Agents}
\label{app:winner}

In this section we give a simple geometric approach for identifying
Bayesian optimal mechanisms for budgeted agents.  This approach can be
used to derive \Cref{thm: opt welfare}, which characterizes the
Bayesian optimal mechanisms for revenue and welfare.  Recall, that the
mechanisms characterized by \Cref{thm: opt welfare} have
deterministic interim payments and are naturally implemented by
all-pay mechanisms.  In this section, we use the approach to analyze
winner-pays-bid mechanisms and identify the Bayesian optimal
mechanisms restricted this family, i.e., we prove
\Cref{thm:first-price}.

The main result of this section is \Cref{lem:opt-win} which describes
the optimal single-agent mechanism for any interim constraint and
combines with the fact that in symmetric single-item environments, the
optimal multi-agent mechanism is given by the interim constraint that
corresponds to the highest-bid-wins allocation rule.
\Cref{thm:first-price}, restated below, follows.

\begin{lemma}{\label{lem:opt-win}}
For a public-budget regular agent and interim allocation constraint
$\calloc(\cdot)$, the welfare-optimal winner-pays-bid single-agent
mechanism allocates as by $\calloc(\cdot)$ except that values in
$[\valp,\highestval]$ are ironed for some $\valp$; payments are given
deterministically by the payment identity and the winner-pays-bid
framework.
\end{lemma}

\begin{numberedtheorem}{\ref{thm:first-price}}
For i.i.d.\ public-budget regular agents, the welfare-optimal
winner-pays-bid mechanism is the highest-bid-wins mechanism, i.e.\ the
first-price auction.
\end{numberedtheorem}

While \Cref{lem:opt-win} can be proven using the traditional analysis of
agents with budgets, e.g., \citet{LR-96}; we will give a proof that
reduces interim optimization to ex ante optimization.  For context,
\citet{BR-89} reduce interim revenue maximization to ex ante revenue
maximization for single-dimensional linear agents 
(i.e.\ without budgets).  
\citet{AFHH-13}
generalized this approach to linear objectives (which do not require
single-dimensional linear agents).  A challenge that our proof
addresses is that public budget agents do not have linear utility
functions and, therefore, welfare maximization is not a linear
objective.

\subsection*{Interim Optimization for Linear Objectives}

We first introduce the {\em interim optimal payoff}, {\em quantile
  space}, and {\em payoff curves}; payoff curves are a straightforward
generalization of revenue curves to non-revenue objectives,
cf.\ \citet{AFHH-13}.

\begin{definition}
\label{d:revenue-linearity}
\label{d:interim-pricing}
For any general objective, the {\em interim optimal payoff}, given
interim allocation constraint $\calloc$, is the payoff of the
single-agent mechanism $(\alloc,\price)$ that satisfies the interim
constraint $\int_\val^{\highval} \alloc(\val)\,d\dist(\val) \leq
\int_\val^{\highval} \calloc(\val)\,d\dist(\val)$ with the highest
objective value; denote this optimal payoff by $\Obj{\calloc}$.  The
objective is {\em linear} if the functional $\Obj{\cdot}$ is linear, i.e., if
for any allocations $\alloc = \alloc\primed + \alloc\doubleprimed$
then $\Obj{\alloc} = \Obj{\alloc\primed} + \Obj{\alloc\doubleprimed}$.
\end{definition}

\begin{definition}
The \emph{quantile} $q$ of a single-dimensional agent with value $\val$
drawn from distribution $\dist$ is the measure with respect to $\dist$
of stronger values, i.e., $\quant = 1 - \dist(\val)$; the
\emph{inverse demand curve} maps an agent's quantile to her value,
i.e., $\val(\quant) = \dist^{-1}(1 - \quant)$.
\end{definition}

\begin{definition}
\label{d:exante-pricing}
For given ex ante allocation probability $\quant$, the single-agent
\emph{ex ante pricing problem} is to find the optimal mechanism with
ex ante allocation probability exactly $\quant$.  The optimal ex ante
payoff, as a function of $\quant$, is denoted by the \emph{payoff
  curve} $\icumvirt(\quant)$.
\end{definition}

The ex ante pricing problem (\Cref{d:exante-pricing}) can be solved
via the following geometric approach.

\begin{definition}
For any $\quant$, denote the {\em price-posting payoff curve},
from posting
price $\val(\quant)$, by $\cumvirt(\quant)$ (with ties broken in favor
of higher payoff).
\end{definition}

If the price-posting payoff curve is differentiable and concave, the
marginal price-posting payoff, a.k.a., the derivative
$\cumvirt'(\quant)$, is well defined and monotone non-increasing.  Its
pointwise optimization leads to an optimal incentive compatible
mechanism, cf.\ \citet{mye-81}.  Specifically, the ex ante pricing
problem is solved by posted pricing and $\icumvirt = \cumvirt$.
Otherwise, analogous to the ironing method of \citet{mye-81}, the ex
ante pricing problem is solved by ironing the price-posting revenue
curve.  

\begin{lemma} [\citealp{AFHH-13}]
Given any linear payoff objective $\Obj{\cdot}$, the payoff curve
$\icumvirt$, which gives the optimal ex ante pricing as a function of
quantile, is given by the concave hull of the price-posting payoff
curve $\cumvirt$.
\end{lemma}

For linear objectives, the interim optimization problem
(\Cref{d:interim-pricing}) is solved by reduction to the ex ante
optimization problem (\Cref{d:exante-pricing}).

\begin{lemma}[\citealp{AFHH-13}]\label{lem:virtual surplus}
Given any linear payoff objective $\Obj{\cdot}$, for any monotone allocation
$\calloc(\cdot)$ and an agent with any price-posting payoff curve $\cumvirt(\quant)$,
the expected payoff of agent is upper-bounded her expected marginal
payoff of the same allocation rule, i.e.,
$$\Obj{\calloc} \leq \expect{\icumvirt'(\quant)\cdot \calloc(\val(\quant))}.
$$ Furthermore, this inequality holds with equality if the allocation
rule $\calloc$ is constant all intervals of values $\val(\quant)$
where $\ironed\cumvirt(\quant) > \cumvirt(\quant)$.
\end{lemma}

This framework allows optimal mechanisms for linear objectives to be
characterized.  The ex ante optimal mechanism is given by an
appropriate ironing of the price-posting payoff curve.  The optimal
interim allocation, for any interim allocation constraint, is given by
ironing the same quantiles as the price-posting payoff curve is
ironed. The resulting allocation rule $\alloc$ optimizes $\icumvirt$
pointwise and is constant on intervals of values $\val(\quant)$ where
$\icumvirt(\quant) > \cumvirt(\quant)$.

\subsection*{Interim Optimization for Lagrangian Objectives}

The difficulty of applying the framework described above is that
welfare maximization (also revenue maximization) with budgeted agents
is non-linear (Definition \ref{d:revenue-linearity}). Our approach prove
Lemma \ref{lem:opt-win} is to consider the optimization program for
welfare maximization with budgeted agents, Lagrangian relax the budget
constraint, and observe that the resulting Lagrangian objective is
linear.  
Then, a characterization of the form of the optimal mechanism for
any Lagrangian objective implies the form of the optimal mechanism for
the optimal choice of the Lagrangian parameter.  Thus, the framework
above for optimizing linear objectives can be effectively applied to
solve the problem of budgeted agents.

We begin by writing an optimization program for the interim welfare
maximization problem for winner-pays-bid mechanisms. The budget
constraint for the winner-pays-bid mechanisms is $\price(\val) \leq
\budget\,\alloc(\val)$ for all $\val \in [0, \highestval]$ where
$\highestval$ is the largest value in the support of distribution $F$.
We introduce the following lemma to help simplify the budget constraint.

\begin{lemma}\label{lem:monotone-ratio}
Given any interim allocation and payment rule $(\alloc, \price)$ in BNE
with $\alloc(\val) > 0$ for $v\in (\valp, \highestval]$, 
then the ratio $\frac{\price(\val)}{\alloc(\val)}$ is 
non-decreasing on $(\valp, \highestval]$.
\end{lemma}
\begin{proof}
By the \citet{mye-81} characterization of BNE, 
the interim allocation rule is non-decreasing 
and the payment rule satisfies the payment identity, i.e.,
$\price(\val) =  \val\cdot \alloc(\val) - \int_0^\val \alloc(t)dt$.
Suppose $\alloc(\val) > 0$ for $\val \in (\valp, \highestval]$.
Consider the derivative of the ratio $\frac{\price(\val)}{\alloc(\val)}$
with respect to $\val$,
\begin{align*}
\frac{d}{d\val}\left[\frac{\price(\val)}{\alloc(\val)}\right]
=
\frac{d}{d\val}\left[\frac{\val\cdot \alloc(\val) - \int_0^\val \alloc(t)dt}{\alloc(\val)}\right]
=
1 - 1 + \frac{\int_0^\val \alloc(t)dt}{(\alloc(\val))^2}\alloc'(\val)
\geq 0.
\end{align*}
Thus, the ratio is non-decreasing in value $\val$ from $\valp$ to $\highestval$.
\end{proof}
Combining Lemma \ref{lem:monotone-ratio} and 
the fact that the budget constraint holds at value $\val$ 
automatically if the allocation $\alloc(\val)$ equal to zero, 
the budget constraint can be simplified
as $\price(\highestval)\leq \budget\,\alloc(\highestval)$.  
We use
Lagrangian relaxation to move the budget constraint into the
objective, and get a optimization program as follows,
\begin{align}\label{prog}
\begin{array}{ll}
\max\limits_{(\alloc,\price)} 
&\expect{\val \cdot \alloc(\val)}
+
\lagrange \budget \alloc(\highestval)- \lagrange \price(\highestval)
\\
s.t.&(\alloc,\price) \text{ are BIC, IIR},
\\
&\text{and
feasible}.
\text{   }
\end{array}
\end{align}

To solve this Lagrangian relaxation program, there will be a correct
Lagrangian multiplier $\lagrange$ for which the budget constraint is
met with equality.  Once the correct value of the Lagrangian
multiplier is determined, the welfare-optimal mechanism can be solved
using the Lagrangian payoff curve that corresponds to the Lagrangian
welfare $\expect{\val \cdot \alloc(\val)}
+\lagrange\budget\alloc(\highestval)-\lagrange\price(\highestval)$.
To prove Lemma~\ref{lem:opt-win}, we give a simple characterization of
the optimal mechanism for any Lagrangian parameter.

Notice that both surplus, revenue, and the allocation and payments of
particular agents are linear objectives.  Thus, for a fixed Lagrangian
parameter the Lagrangian welfare optimization is a linear objective.

\begin{lemma}
Given any interim allocation and payment rule 
$(\alloc,\price)$ in BNE,
fox any fixed Lagrangian parameter $\lagrange$,
the
Lagrangian welfare $\expect{\val \cdot \alloc(\val)}
+\lagrange\budget\alloc(\highestval)-\lagrange\price(\highestval)$ 
is a linear
objective.
\end{lemma}
\begin{proof}
By the definition of linearity,
the sum of linear objectives 
and 
a scalar multiple of a linear objective are both 
linear objectives.
$\alloc(\val)$ is linear for all $\val$,
hence, the scalar multiple of allocation 
$\lambda\budget\alloc(\highestval)$, 
and 
the expected surplus 
$\expect{\val\cdot\alloc(\val)}$
are both linear objectives.
By the payment identity,
the payment is an integral of $\alloc(v)$
and integral is a linear operator,
hence,
the scalar multiple of price 
$\lambda\price(\highestval)$
is also a linear objective.
Thus, the 
Lagrangian welfare $\expect{\val \cdot \alloc(\val)}
+\lagrange\budget\alloc(\highestval)-\lagrange\price(\highestval)$ 
is a linear
objective.
\end{proof}

\begin{figure}[t]
\centering
\input{winner-figure}
\caption{\label{f:curve} 
Depicted are the
 $-\lagrange\val(\quant)$ curve (black dashed), the
 $\int_{0}^{\quant}\val(t)dt +\lagrange\budget$ curve (black dotted),
 the Lagrangian welfare curve $\cumval_\lagrange(\quant)$ (gray solid),
 and the ironed Lagrangian welfare curve $\ironed\cumval_\lagrange(\quant)$ (black solid)
 corresponding to an agent with value uniformly distributed on $[0,1]$ with budget $\budget = \tfrac{1}{4}$ and Lagrangian multiplier $\lagrange = \frac{2}{5}$.
  }
\end{figure}
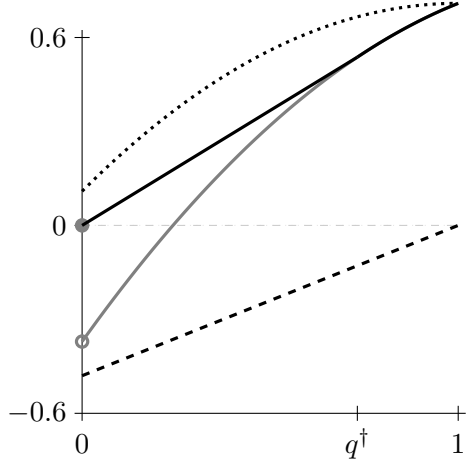

Consider optimizing the program \eqref{prog} for a fixed Lagrangian parameter
$\lagrange$.  To apply the framework discussed previously, we first
construct the price-posting payoff curve.  Notice that the identified
Lagrangian price-posting welfare curve is discontinuous at $\quant =
0$ (unless $\lagrange = 0$, i.e.\ when the budget constraint is not
binding).
\begin{lemma}
The Lagrangian price-posting welfare curve $\cumval_\lagrange(\cdot)$
for a public budget agent satisfies
\begin{align*}
\cumval_{\lagrange}(\quant) = 
\left\{
\begin{array}{ll}
0 \quad &\text{if $\quant = 0$},
\\
\int_{0}^{\quant} \val(\quant)d\quant - \lagrange\val(\quant) + \lagrange\budget
\quad &
\text{otherwise.}
\end{array}
\right.
\end{align*}
\end{lemma}

\begin{proof}
Consider posting a price $\val(\quant)$ in the quantile space. 

For $\quant > 0$ (strictly positive), the price $\val(\quant)$ is
strictly less than the highest value $\val(0) = \highestval$, so
$\price(\highestval)=\val(\quant)$ and
$\alloc(\highestval) = 1$.  Thus, the Lagrangian objective
for $\quant \in (0, 1]$ is $\cumval_{\lagrange}(\quant) =
  \int_{0}^{\quant}\val(t)dt - \lagrange \val(\quant) +
  \lambda\budget$.

For $\quant = 0$, an agent with the highest value $\highestval$ is indifferent
between buying and not buying.  
If this agent buys, the objective is negative; 
if this
agent does not buy, the objective is zero. 
Per the definition of the
price-posting revenue curve, we break this tie in favor of the
objective.   
\end{proof}

\begin{proof}[Proof of Lemma~\ref{lem:opt-win}]
If the budget does not bind, it is optimal to allocate  
the interim constraint $\calloc(\cdot)$, since the welfare of any interim feasible allocation is at most the same as the welfare of 
the interim constraint.

Next, we assume that the budget binds, i.e., $\lagrange > 0$.  Notice
that on $\quant\in (0, 1]$, the Lagrangian price-posting welfare curve
is the constant $\lagrange\budget$ plus the difference between the
original welfare curve $\int_{0}^{\quant}\val(t)dt$ and the scaled
value function $\lagrange\val(\quant)$.  
Since the original welfare
curve is always concave, 
under the public-budget regularity
assumption, 
this Lagrangian price-posting welfare curve is concave
on $\quant \in (0, 1]$.  
Due to the discontinuity at $\quant = 0$,
and the fact that for sufficient small $\quant > 0$, $\cumval_\lagrange(\quant)$ is negative, 
the Lagrangian welfare curve (i.e., solving the ex ante
optimization problem) is the Lagrangian price-posting welfare
curve with ironing from $0$ to some $\quant\primed$ (See
Figure~\ref{f:curve}).  
Therefore, by Lemma \ref{lem:virtual surplus}, the welfare-optimal 
mechanism allocates as by the interim constraint $\calloc(\cdot)$ except ironing the top values between $\valp =
\val(\quant\primed)$ and $\highestval = \val(0)$.  
\end{proof}

%% file: winner-figure.tex
\begin{tikzpicture}[scale = 0.5]

\draw [color = gray!20!white, dotted] (0, 5) -- (10, 5);

\draw (-0.2,0) -- (10.2, 0);
\draw (0, -0.2) -- (0, 10.2);
\draw (7.31825, 0.2) -- (7.31825, -0.2);
\draw (10, 0.2) -- (10, -0.2);
\draw (-0.2, 5) -- (0.2, 5);
\draw (-0.2, 10) -- (0.2, 10);

\draw [color = gray!40!white, dashed] (0, 5) -- (10, 5);

\begin{scope}[very thick]

\draw[color = gray] (0, 1.90911) circle (0.15cm);
\draw[color = gray, fill=gray] (0, 5) circle (0.15cm);

\draw[scale = 9.091, domain = 0:1.1, smooth, variable=\x, dotted]
plot({\x},{(\x - 0.5 *\x *\x / 1.1 + 0.65)});

\draw[scale = 9.091, domain = 0:1.1, smooth, variable=\x, dashed]
plot({\x},{-0.4 * (1.1-\x) + 0.55 });

\draw[scale = 9.091, domain = 0:1.1, smooth, variable=\x, color = gray]
plot({\x},{(\x - 0.5 *\x *\x / 1.1 + 0.1)-0.4 * (1.1-\x) + 0.55 });

\draw[scale = 9.091, domain = 0.8050075:1.1, smooth, variable=\x, color = black]
plot({\x},{(\x - 0.5 *\x *\x / 1.1 + 0.1)-0.4 * (1.1-\x) + 0.55 });

\draw (0, 5) -- (7.31825, 9.476896523983468);



\end{scope}

\draw (0, -0.8) node {$0$};
\draw (7.31825, -0.8) node {$\quant\primed$};
\draw (10, -0.8) node {$1$};

\draw (-1.2, 0) node {$-0.6$};
\draw (-0.6, 5) node {$0$};
\draw (-0.9, 10) node {$0.6$};

\end{tikzpicture}